\documentclass[hyperref,amsthm,final]{elsart}

\usepackage[utf8]{inputenc}

\def\coloneq{:=}

\usepackage[T1]{fontenc}

\usepackage{yjsco}

\usepackage{etex}
\usepackage{etoolbox}

\usepackage{xspace}

\usepackage[autostyle]{csquotes}

\usepackage{amsmath}
\usepackage{amsthm}

\usepackage{amssymb}

\usepackage{dsfont} 

\usepackage{mathtools}

\mathtoolsset{%
  showonlyrefs,
  showmanualtags}

\usepackage{array}
\usepackage[binary-units=true]{siunitx}

\usepackage{manfnt} 

\usepackage{verbatim}

\usepackage{eqparbox}

\usepackage{letltxmacro}

\usepackage{graphicx}

\usepackage{algorithm}
\usepackage[noend]{algorithmic}

\makeatletter
\renewcommand{\listofalgorithms}{
  \@ifundefined{ext@algorithm}{\float@error{algorithm}}{%
    \@namedef{l@algorithm}{\@dottedtocline{1}{1.5em}{2.3em}}%
    \float@listhead{\listalgorithmname}%
    \begingroup\setlength{\parskip}{\z@}%
    \@starttoc{\@nameuse{ext@algorithm}}%
    \endgroup}}
\makeatother

\usepackage{booktabs}

\usepackage{multirow}
\usepackage{tablefootnote}

\newcommand{\pbox}[2][l]{%
  \begin{tabular}[c]{@{}#1@{}}#2\end{tabular}%
}

\usepackage{tikz}
\usetikzlibrary{arrows,calc,3d,matrix,intersections,external,fadings,patterns,decorations.pathreplacing,positioning}

\usepackage{ifdraft}

\usepackage{fixme}

\usepackage[normalem]{ulem}

\usepackage{enumitem} 
\usepackage{subcaption}

\usepackage[frenchlinks,colorlinks,linkcolor=blue,final,debug]{hyperref}

\newlist{equivlist}{enumerate}{2}
\setlist[equivlist]{label=(\arabic*),ref=(\arabic*)}

\newcommand{\N}{\mathbb{N}}

\newcommand{\K}{\mathbb{K}}
\newcommand{\KK}{\mathbb{K}}

\newcommand{\x}{\times} 

\newcommand{\ifnonempty}[3]{%
  \def\tempa{}%
  \def\tempb{#1}%
  \ifx\tempa\tempb 
  #3 
  \else            
  #2
  \fi}

\newcommand{\divides}{\mid}

\DeclareMathOperator*{\toAux}{\longrightarrow}
\renewcommand{\to}{\toAux\limits}

\renewcommand{\mapsto}{\longmapsto}

\renewcommand{\bar}[1]{\overline{#1}}

\renewcommand{\epsilon}{\varepsilon}
\renewcommand{\phi}{\varphi}

\floatname{algorithm}{Algorithme}
\renewcommand{\listalgorithmname}{Liste des algorithmes}

\usepackage[english]{babel}

\theoremstyle{plain} 
\newtheorem{theorem}{Theorem}
\newtheorem*{theorem*}{Theorem}
\newtheorem{lemma}[theorem]{Lemma}
\theoremstyle{definition}
\newtheorem{definition}[theorem]{Definition}
\newtheorem{conjecture}[theorem]{Conjecture}

\theoremstyle{remark}

\newtheorem*{rmk}{Remark}
\newtheorem*{remark}{Remark}

\newcommand{\for}{\text{ for }}

\newcommand{\textin}{\text{ in }}

\newcommand{\textif}{\text{ if }} 
\newcommand{\otherwise}{\text{ otherwise }}

\newcommand{\WOne}{\mathds{1}}

\renewcommand{\hom}[2][]{
\hspace{0.1em}%
\mathrm{hom}_{#2}^{\ifnonempty{#1}{#1}{\phantom{-1}}}
\ifnonempty{#1}{\hspace{-0.1em}}{\hspace{-0.2em}}%
}

\newcommand{\F}[1]{
\ifmmode%
        \mathsf{F_{#1}}%
\else%
        \textsf{F}$_{\mathsf#1}$%
\fi\xspace}

\newcommand{\grevlex}{\textsc{GRevLex} }
\newcommand{\lex}{\textsc{Lex} }
\newcommand{\Wgrevlex}{$W$-\grevlex}

\newcommand{\ltgrevlex}{<_\text{grevlex}}

\newcommand{\ltWgrevlex}{<_\text{$W$-grevlex}}

\newcommand{\LT}{\mathsf{LT}}

\newcommand{\Spol}{S\text{-}\mathrm{pol}}

\newcommand{\HS}{\mathsf{HS}}

\newcommand{\HI}{i_{\mathrm{reg}}}
\newcommand{\dreg}[1][]{d_{\mathrm{reg}\ifnonempty{#1}{,#1}{}}}
\newcommand{\dmax}{d_\mathrm{max}}

\renewcommand{\mod}[1]{\,[#1]}

\newcommand{\ceil}[1]{\lceil#1\rceil}
\newcommand{\floor}[1]{\lfloor#1\rfloor}

\newcommand{\superscript}[1]{\ensuremath{^{\textrm{#1}}}}

\newenvironment{eqlist}[1][]{\enumerate[itemindent=*,%
  labelindent=\parindent,label=\textup{\textsf{(#1\arabic*)}},
  ref=#1\arabic*]}%
{\endenumerate}

\fxsetup{author={Thibaut},lang=french,envlayout={color},%
  layout={marginclue,footnote},theme=colorsig}
\FXRegisterAuthor{ms}{anms}{Mohab}
\FXRegisterAuthor{jc}{anjc}{Jean-Charles}

\hypersetup{%
  pdfinfo={
    Title={On the complexity of computing Gröbner bases for weighted
      homogeneous systems},
    Author={} } }

\usepackage{natbib}

\renewcommand{\mod}{\,\mathrm{mod}\,}
\newcommand{\Fext}{F_{\mathrm{ext}}}

\newcommand{\order}[1]{<_{#1}}

\newcommand{\Ocomp}[1]{\ensuremath{\operatorname{O}\left(#1\right)}}
\newcommand{\ocomp}[1]{\ensuremath{\operatorname{o}\left(#1\right)}}

\newcommand{\FGLM}{\textsf{FGLM}\xspace}

\tikzset{dot/.style={draw, circle, inner sep=1pt, fill}}
\tikzset{seriesLine/.style={draw}}
\tikzset{axis/.style={->,thick}}

\newcommand{\drawSeries}[5]
{%
  \xdef\xold{#3} \xdef\yold{#4} \foreach \x / \y
  [
  evaluate=\x as \xmin using {(\yold <= \y)*(\x) + (\yold >
    \y)*(\xold)}, evaluate=\y as \ymin using {min(\yold,\y)}, count=\i
  from #5] in #1 {%
    \node[dot] (#2\i ) at (\x ,\y ) {} ; \draw[seriesLine] (\xold
    ,\yold ) -- (\xmin,\ymin) -- (\x ,\y ); \xdef\xold{\x}
    \xdef\yold{\y} } }

\begin{document}

\begin{frontmatter}
  \title{On the complexity of computing Gröbner bases for weighted
    homogeneous systems}

  \def\labo{\superscript{a,b,c}}
  \def\iuf{\superscript{d}}
  \def\laboiuf{\superscript{a,b,c,d}}

  \author[upmc,cnrs,inria]{Jean-Charles Faugère}
  \ead{Jean-Charles.Faugere@inria.fr},
  \author[upmc,cnrs,inria,iuf]{Mohab Safey El Din}
  \ead{Mohab.Safey@lip6.fr},
  \author[upmc,cnrs,inria]{Thibaut Verron}
  \ead{Thibaut.Verron@lip6.fr}

  \address[upmc]{Sorbonne Universités, UPMC Univ Paris 06, 7606, LIP6,
    F-75005, Paris, France}
  \address[cnrs]{CNRS, UMR 7606, LIP6, F-75005, Paris, France}
  \address[inria]{Inria, Paris-Rocquencourt Center, PolSys Project}
  \address[iuf]{Institut Universitaire de France}


  \begin{abstract}
    Solving polynomial systems arising from applications is frequently
    made easier by the structure of the systems. Weighted homogeneity
    (or quasi-homogeneity) is one example of such a structure: given a
    system of weights $W=(w_{1},\dots,w_{n})$, $W$-homogeneous
    polynomials are polynomials which are homogeneous w.r.t the weighted
    degree
    $\deg_{W}(X_{1}^{\alpha_{1}}\dots X_{n}^{\alpha_{n}}) = \sum
    w_{i}\alpha_{i}$.
    
    Gröbner bases for weighted homogeneous systems can be computed by
    adapting existing algorithms for homogeneous systems to the weighted
    homogeneous case.  We show that in this case, the complexity
    estimate for Algorithm~\F5
    $\left(\binom{n+\dmax-1}{\dmax}^{\omega}\right)$ can be divided by a
    factor $\left( \prod w_{i} \right)^{\omega}$. For zero-dimensional
    systems, the complexity of Algorithm~\FGLM $nD^{\omega}$ (where $D$
    is the number of solutions of the system) can be divided by the same
    factor $\left( \prod w_{i} \right)^{\omega}$.
    Under genericity
    assumptions, for zero-dimensional weighted homogeneous systems of
    $W$-degree $(d_{1},\dots,d_{n})$, these complexity estimates are polynomial in
    the weighted Bézout bound
    $\prod_{i=1}^{n}d_{i} / \prod_{i=1}^{n}w_{i}$.

    Furthermore, the maximum degree reached in a run of Algorithm \F5 is
    bounded by the weighted Macaulay bound $\sum (d_{i}-w_{i}) + w_{n}$,
    and this bound is sharp if we can order the weights so that
    $w_{n}=1$. For overdetermined semi-regular systems, estimates from
    the homogeneous case can be adapted to the weighted case.

    We provide some experimental results based on systems arising from a
    cryptography problem and from polynomial inversion problems. They
    show that taking advantage of the weighted homogeneous structure
    can yield substantial speed-ups, and allows us to solve systems which
    were otherwise out of reach.
  \end{abstract}
\end{frontmatter}
\endNoHyper

\section{Introduction}
\label{sec:Introduction}
Algorithms for solving polynomial systems have become increasingly
important over the past years, because of the many situations where
such algebraic systems appear, including both theoretical problems
(algorithmic geometry, polynomial inversion\ldots), and real-life
applications (cryptography, robotics\ldots).
Examples of such algorithms include eigenvalues methods for systems with a finite number of solutions, or resultant calculations for polynomial elimination (see~\citep{DickensteinEmiris2010} for a survey).

The theory of Gröbner bases is another tool which has proved useful for this purpose, and many algorithms for computing Gröbner bases have been described since their introduction.
They include direct algorithms, computing the Gröbner basis of any system: to name only a few, the historical Buchberger algorithm~\citep{Buch76}, and later the Faugère \F4~\citep{F99a} and \F5~\citep{Fau02a} algorithms; as well as change of order algorithms, computing a Gröbner basis of an ideal from another Gröbner basis: the main examples are the \FGLM\ algorithm \citep{FGLM} for systems with a finite number of solutions, and the Gröbner walk \citep{Collart97} for the general case.

Systems arising from applications usually have some structure, which
makes the resolution easier than for generic systems.  In this paper, we
consider one such structure, namely \emph{weighted homogeneous}
polynomials: a polynomial $f(X_{1},\dots,X_{n})$ is weighted
homogeneous with respect to a system of weights
$W=(w_{1},\dots,w_{n})$ (or $W$-homogeneous) if and only if
$f(X_{1}^{w_{1}},\dots,X_{n}^{w_{n}})$ is homogeneous in the usual
sense.

Moreover, in order to obtain precise results, we will assume that
the systems satisfy some \emph{generic} properties, which are
satisfied by almost any system drawn at random.  This is a usual
assumption for Gröbner basis complexity estimates.  More generally, we
will also consider affine systems with a \emph{weighted homogeneous
structure}, that is systems whose component of maximal weighted degree
will satisfy these generic properties.

The complexity estimates given in this paper can be applied to a wide
range of Gröbner basis algorithms. However, we mainly focus on two
algorithms: Matrix-\F5, which is a matrix variant of \F5 described
in~\cite{BFS14}, allowing for complexity analyses, and
\FGLM.

\paragraph*{Prior work}
\label{sec:Previous-work}
The special case $W=(1,\dots,1)$ is the usual homogeneous case.  In
this case, all the results from this paper specialize to known results.
Furthermore, some hypotheses are always satisfied, making the
properties and definitions simpler. In particular, the description of
the Hilbert series of a homogeneous complete intersection is adapted
from~\cite{MS03}, and the asymptotics of the degree of regularity of a
semi-regular sequence were studied in~\cite{BFSY05}.

Weighted homogeneous systems have been studied before, from the angle
of singularity theory and commutative algebra.  In particular,
some results about the Hilbert series and the Hilbert function of weighted
homogeneous ideals, including the weighted Bézout
bound~\eqref{eq:intro-bezout}, can be found in most commutative
algebra textbooks. 

The computational strategy for systems with a weighted structure is not new either, for example it is already implemented (partially: only for weighted homogeneous systems with a degree order) in the computer algebra system Magma \citep{Magma}.
Additionally, the authors of~\cite{Tra1996} proposed another way of taking into account the weighted structure, by way of the Hilbert series of the ideal.
The authors of~\cite{CDR96} generalized this algorithm to systems homogeneous with respect to a multigraduation.
Their definition of a system of weights is more general than the one we use in the present paper.

To the best of our knowledge, nobody presented a formal description of
a computational strategy for systems with a weighted homogeneous
structure (not necessarily weighted homogeneous),
together with complexity estimates. 

Some of the results presented in this paper about regular sequences
previously appeared in a shorter conference paper \citep{FSV13}, of
which this paper is an extended version: these results are the weak form of the
weighted Macaulay bound~\eqref{eq:intro-weak-macaulay} and the formal
description of the algorithmic strategy for weighted homogeneous
systems, with the complexity estimates~\eqref{eq:intro-comp-F5}
and~\eqref{eq:intro-comp-FGLM}. This conference paper lacked a hypothesis
(reverse chain-divisible systems of weights), and as such lacked the
precise description of Hilbert series required to obtain results for
semi-regular sequences. The sharp variant of the weighted Macaulay
bound~\eqref{eq:intro-macaulay}, under the assumption of simultaneous
Noether position, was also added in the present paper. Finally, the
benchmarks section of the current paper contains additional systems,
arising in polynomial inversion problems.

The conference paper was using \emph{quasi-homogeneous} to
describe the studied structure, instead of \emph{weighted
  homogeneous}.  While both names exist in the literature,
\emph{weighted homogeneous} seems to be more common, and to better
convey the notion that this structure is a generalization of
homogenity, instead of an approximation.  The same notion is sometimes
also named simply \emph{homogeneous} (in which case the weights are
determined by the degree of the generators; see for
example~\cite{eisenbud95}), or homogeneous for a \emph{nonstandard
  graduation} \citep{dalzotto2006}.

\paragraph*{Main results}
\label{sec:Main-results}
By definition, weighted homogeneous polynomials can be made
homogeneous by raising all variables to their weight.  The resulting
system can then be solved using algorithms for homogeneous systems.
However, experimentally, it appears that solving such systems is much
faster than generic homogeneous systems.  In this paper, we show that
the complexity estimates for homogeneous systems, in case the system
was originally $W$-homogeneous, can be divided by
${(\prod w_{i})}^{\omega}$, where $\omega$ is the complexity exponent
of linear algebra operations ($\omega = 3$ for naive algorithms, such
as the Gauss algorithm).

These complexity estimates depend on two parameters of the system: its
\emph{degree of regularity} $\dreg$ and its \emph{degree}
$\deg(I)$. These parameters can be obtained from the \emph{Hilbert
  series} of the ideal, which can be precisely described under generic
assumptions.  To be more specific, we will consider systems defined by a
\emph{regular sequence} (Def.~\ref{def:regular}) and systems which are
in \emph{simultaneous Noether position} (Def.~\ref{def:snp}).

\begin{theorem*}
  \label{thm:intro-comp}
  Let $W=(w_{1},\dots,w_{n})$ be a system of weights, and $F=(f_{1},\dots,f_{m})$ a zero-dimensional $W$-homogeneous system of polynomials in $\KK[X_{1},\dots,X_{n}]$, with respective $W$-degree $d_{1},\dots,d_{m}$.
  The complexity (in terms of arithmetic operations in $\KK$) of Algorithm~\F5 to compute a \Wgrevlex Gröbner basis of $I \coloneq \langle F \rangle$ is bounded by
  \begin{equation}
    \label{eq:intro-comp-F5}
    C_{\F5} = \Ocomp{
      \frac{1}{{(\prod w_{i})}^{\omega}}
      \cdot \binom{n+\dreg-1}{\dreg}^{\omega}
    }.
  \end{equation}
  If $F$ is a regular sequence (and in particular $m=n$), then
  $\dreg$ can be bounded by the \emph{weighted Macaulay bound}:
  \begin{equation}
    \label{eq:intro-weak-macaulay}
    \dreg \leq \sum_{i=1}^{n} (d_{i} - w_{i}) + \max\{w_{j}\}.
  \end{equation}
  If additionally $F$ is in simultaneous Noether position w.r.t
  the order $X_{1} > \dots > X_{n}$, then the weighted Macaulay bound
  can be refined:
  \begin{equation}
    \label{eq:intro-macaulay}
    \dreg \leq \sum_{i=1}^{n} (d_{i} - w_{i}) + w_{n}.
  \end{equation}

  The complexity of Algorithm~\FGLM to perform a change of ordering is
  bounded by
  \begin{equation}
    \label{eq:intro-comp-FGLM}
    C_{\FGLM} = \Ocomp{n {\left(\mathrm{deg(I)}\right)}^{\omega}}.
  \end{equation}
  If $F$ forms a regular sequence, then $\deg(I)$ is given by the
  \emph{weighted Bézout bound}
  \begin{equation}
    \label{eq:intro-bezout}
    \deg(I) = \frac{\prod_{i=1}^{n}d_{i}}{\prod_{i=1}^{n}w_{i}}.
  \end{equation}
\end{theorem*}

In particular, the bound~\eqref{eq:intro-macaulay} indicates that in
order to compute a Gröbner basis faster for a generic enough system,
one should order the variables by decreasing weights whenever
possible.

The hypotheses of the theorem are not too restrictive. In the
homogeneous case, regularity and simultaneous Noether position are
generic properties.  However, in the weighted homogeneous case, there
are systems of weights and systems of weighted degrees for which they
are not generic.  In this paper, we identify large families of systems
of weights and systems of weighted degrees for which they are
(Prop.~\ref{thm:gen-2}).

All sequences in simultaneous Noether position are regular.  In the
homogeneous case, conversely, all regular sequences are in
simultaneous Noether position up to a generic linear change of
coordinates.  In the weighted homogeneous case, it is no longer true.
Worse still, there are systems of weights for which there exists no
non-trivial change of coordinates.

In order to work around this limitation, we consider \emph{reverse
  chain-divisible} systems of weights, that is systems of weights such
that $w_{n} \divides w_{n-1} \divides \dots \divides w_{1}$.  This
property ensures that there are non-trivial change of coordinates of
the form $X_{i} \leftarrow X_{i} + P_{i}(X_{i+1},\dots,X_{n})$ for all
$i$, with $P_{i}$ a $W$-homogeneous polynomial with $W$-degree $w_{i}$.
Under this assumption, many properties from the homogeneous case
remain valid in a weighted setting, and in particular, any regular
sequence is, up to a $W$-homogeneous change of coordinates, in
simultaneous Noether position (Th.~\ref{thm:gen-linchg}).

For many systems from practical applications, the weights can be chosen to be reverse chain-divisible.
We give a few examples in the last section of this paper.


If $m>n$, there is no regular sequence.
Instead, we will consider systems defined by a \emph{semi-regular} sequence, that is systems for which no reduction to zero appear in a run of Algorithm~\F5.
This property has several equivalent definitions in the homogeneous case.
While these definitions can be easily extended to the weighted case, their equivalence is not necessarily true.
However, we prove that these definitions are equivalent in the special case where the weights form a reverse chain-divisible sequence.

In the homogeneous case, the property of being semi-regular is only
conjectured to be generic, but this conjecture is proved in a handful
of cases~(\citet[Thm.~1.5]{MS96}.  In this paper, we adapt the proof of
one of these cases, namely the case $m=n+1$ in a base field of
characteristic~$0$.

For semi-regular systems with $m=n+1$, we obtain a bound on the degree
of regularity of the system. More generally, in the homogeneous case,
one can compute asymptotic estimates on the degree of regularity of a
semi-regular sequence \citep{BFSY05,Bar04}.  These estimates can be
adapted to the weighted homogeneous case. As an example, we give an
asymptotic bound on the degree of regularity for semi-regular systems
with $m=n+k$ for a given integer $k$:
\begin{theorem*}
  \label{thm:intro-dreg}
  Let $n$ and $k$ be two positive integers, and let $m=n+k$.  Let
  $w_{0}$ and $d_{0}$ be two positive integers such that
  $w_{0} \divides d_{0}$.  Consider the system of $n$ weights
  $W = (w_{0},\dots,w_{0},1)$.  Let $F$ be a semi-regular sequence in
  $\KK[X_{1},\dots,X_{n}]$, made of $W$-homogeneous polynomials with
  $W$-degree $d_{0}$.  Then the highest degree reached in the
  computation of a \Wgrevlex Gröbner basis of $\langle F \rangle$ is
  asymptotically bounded by
  \begin{equation}
    \label{eq:intro-dreg-asymp}
    \dreg
    = n \,\frac{d_{0}-w_{0}}{2}
    - \alpha_{k}\sqrt{n\, \frac{d_{0}^{2}-w_{0}^{2}}{6}}
    + \Ocomp{n^{1/4}}.
  \end{equation}
  where $\alpha_{k}$ is the largest root of the $k$'th Hermite's
  polynomial.
\end{theorem*}

Experimentally, if we lift the assumption that the system of weights is reverse chain-divisible, the degree of regularity does not appear to rise too far beyond the bound.
Future work on the topic could include characterizing the Hilbert series of $W$-homogeneous semi-regular sequences in full generality, in order to obtain bounds on the $W$-degree of regularity.

In practice, taking advantage of the weighted structure when applicable yields significant speed-ups.
Some instance of a weighted structure has already been successfully exploited for an application in cryptography \citep{FGHR13}.
We also present timings obtained with several polynomial inversion problems, with speed-ups ranging from 1--2 to almost 100.
In particular, we use these techniques in order to compute the relations between fundamental invariants of several groups (see \citep{Sturmfels2008}).
For some groups such as the Cyclic-5 group or the dihedral group $D_{5}$, computing these relations is intractable without considering the weighted structure of the system, while it takes only a few seconds or minutes when exploiting the weighted structure.
All these systems are examples of applications where the weights giving the appropriate $W$-homogeneous structure are naturally reverse chain-divisible.
These experimentations have been carried using \F5 and \FGLM with the Gröbner basis library FGb~\citep{F10c} and \F4 with the computer algebra system Magma~\citep{Magma}.

There are other applications where Gröbner bases are computed for polynomial systems with a weighted-homogeneous structure, for example in coding theory, both for generating codes (\citet[sec.~5]{deBoerPellikaan1999}, \citep{Leonard2009}) and for decoding through Guruswami-Sudan's algorithm (see \citep{GuerriniRimoldi2009} for an overview).

\paragraph*{Organisation of the paper}
\label{sec:Organisation-paper}

In section~\ref{sec:Notat-prev-results}, we define weighted graded
algebras and some generic properties of weighted homogeneous systems. In
section~\ref{sec:Regular-systems}, we focus on regular systems and
complete intersections. We describe the Hilbert series of a weighted
homogeneous complete intersection and give the sharp variant of the
weighted Macaulay bound. In section~\ref{sec:Semi-regular-systems}, we
consider semi-regular systems. We give some equivalent definitions of
this property, and we show how asymptotic estimates of the degree of
regularity can be adapted from the homogeneous case to the weighted
case. Additionally, we prove that Fröberg's conjecture in the case
$m=n+1$ is true in the weighted case, as in the homogeneous case,
provided that the base field is large enough. In
section~\ref{sec:Cons-comp}, we describe strategies for computing
Gröbner bases for weighted homogeneous systems, and we give complexity
estimates for these strategies. Finally, in
section~\ref{sec:Applications}, we show how weighted structures can
appear in applications, and we give some benchmarks for each example.
 
\section{Definitions and genericity statements}
\label{sec:Notat-prev-results}

\subsection{Definitions}
\label{sec:Definitions}

Let $\K$ be a field.  We consider the algebra
$\K[X_{1},\dots,X_{n}] = \K[\mathbf{X}]$.  This algebra can be graded
with respect to a system of weights, as seen for example
in~\cite[sec.~10.2]{becker1993grobner}.

\begin{definition}
  Let $W = (w_{1},\dots,w_{n})$ be a vector of positive integers.
  Let $\alpha = (\alpha_{1}, \dots, \alpha_{n})$ be a vector of nonnegative integers.
  Let the integer $ \deg_{W}(\mathbf{X^{\alpha}}) = \sum_{i=1}^{n} w_{i}\alpha_{i}$ be the \hbox{\emph{$W$-degree}}, or \emph{weighted degree} of the monomial $\mathbf{X}^{\alpha} = X_{1}^{\alpha_{1}}\cdots X_{n}^{\alpha_{n}}$.
  We say that the vector $W$ is a \emph{system of weights}.
  We denote by $\mathbf{1}$ the system of weights defined by $(1,\dots,1)$, associated with the usual grading (in total degree) on $\KK[\mathbf{X}]$.
\end{definition}

Any grading on $\K[\mathbf{X}]$ comes from such a system of weights
\citep[sec.~10.2]{becker1993grobner}.  When working with a
$W$-graduation, to clear up any ambiguity, we use the adjective
\emph{$W$-homogeneous} for elements or ideals, or \emph{weighted
  homogeneous} if $W$ is clear in the context.  The word
\emph{homogeneous} will be reserved for $\mathbf{1}$-homogeneous
items.  The following property is an easy consequence of the
definition.

\begin{prop}
  Let $(\K[X_{1},\dots,X_{n}],W)$ be a graded polynomial algebra.
  Then the application
  \begin{equation}
    \label{eq:3}
    \begin{matrix}
      \hom{W} : & (\K[X_{1},\dots,X_{n}],W) & \to & (\K[t_{1},\dots,t_{n}],\mathbf{1}) \\
      & f & \mapsto & f(t_{1}^{w_{1}},\dots,t_{n}^{w_{n}})
    \end{matrix}
  \end{equation}
  is an injective graded morphism, and in particular the image of a
  weighted homogeneous polynomial is a homogeneous polynomial.
\end{prop}

The above morphism also provides a weighted variant of the \grevlex
ordering (as found for example in~\cite[10.2]{becker1993grobner}), called
the \Wgrevlex\ ordering:
\begin{equation}
  \label{eq:102}
  u \ltWgrevlex v
  \iff \hom{W}(u) \ltgrevlex \hom{W}(v).
\end{equation}
Given a $W$-homogeneous system $F$, one can build the homogeneous
system $\hom{W}(F)$, and then apply classical algorithms
\citep{Fau02a,FGLM} to that system to compute a \grevlex\
(resp. \lex{}) Gröbner basis of the ideal generated by $\hom{W}(F)$.

\begin{definition}
  The \emph{$W$-degree of regularity} of the system $F$ is the highest
  degree $\dreg[W](F)$ reached in a run of \F5 to compute a \grevlex{}
  Gröbner basis of $\hom{W}(F)$. When the graduation is clear in the
  context, we may call it degree of regularity, and denote it $\dreg$.
\end{definition}
\begin{rmk}
  Unlike what we could observe in the homogeneous case, this definition depends on the
  order of the variables (we shall give an example in Table~\ref{tab:dreg} in section~\ref{sec:Degr-regul-quasi}, and another, with timings, in Table~\ref{tab:dreg-timings} in section~\ref{sec:Generic-systems}).%
\end{rmk}

\begin{definition}
  Let $I$ be a zero-dimensional (not necessarily weighted homogeneous)
  ideal in $A=\KK[X_{1},\dots,X_{n}]$.  In that case, we define the
  degree $D$ of the ideal $I$ as the (finite) dimension of $A/I$, seen
  as a $\KK$-vector space:
  \begin{equation}
    \label{eq:101}
    D = \dim_{\KK}\left( A/I \right).
  \end{equation}
  Equivalently, if $\HS_{A/I}(T)$ is the Hilbert series (with respect
  to the $W$-graduation) of $I$, this
  series is a polynomial in $T$ and
  \begin{equation}
    \label{eq:75}
    D=\HS_{A/I}(1).
  \end{equation}
\end{definition}
\begin{rmk}
  This definition with the Hilbert series can be extended to ideals with positive dimension.
  However, in a weighted setup, varieties can end up having rational (not-necessarily integer) degrees.
  This is the definition used by the software Macaulay2 \citep[function \texttt{degree(Module)}]{Macaulay2}.
\end{rmk}

We will only consider the \emph{affine} varieties associated with the
ideals we consider.  In particular, the dimension of $V(0)$ is $n$,
and a zero-dimensional variety is defined by at least $n$ polynomials
if the base field is algebraically closed.

\begin{definition}[Regular sequence]
  \label{def:regular}
  Let $W=(w_{1},\dots,w_{n})$ be a system of weights, let \hbox{$D=(d_{1},\dots,d_{m})$} be a system of \hbox{$W$-degrees} and let $F=(f_{1},\dots,f_{m})$ be a sequence of $W$-homogeneous polynomials in $\KK[X_{1},\dots,X_{n}]$, with $W$-degree $D$.
  The system $F$ is called \emph{regular} if it satisfies one of the following equivalent properties \citep{eisenbud95}:
  \begin{enumerate}
    \item
    $\forall\, i \in \{1, \dots, m\},\, f_{i} \text{ is not a
      zero-divisor in } \KK[X_{1},\dots,X_{n}]/\langle f_{1},\dots,
    f_{i-1} \rangle$;
    \item the Hilbert series of $\langle F \rangle$ is given by
    \begin{equation}
      \label{eq:HSreg}
      \HS_{A/I}(T)=\frac{\prod_{i=1}^{m}(1-T^{d_{i}})}
                        {\prod_{i=1}^{n}(1-T^{w_{i}})}.
    \end{equation}
  \end{enumerate}
\end{definition}

\begin{definition}[Simultaneous Noether position]
  \label{def:snp}
  Let $W$ be a system of weights.  Let $m \leq n$ and
  $F = (f_{1},\dots,f_{m})$ be a sequence of $W$-homogeneous polynomials
  in $\KK[X_{1},\dots,X_{n}]$ The system $F$ is said to be \emph{in
    Noether position w.r.t the variables $X_{1},\dots,X_{m}$} if it
  satisfies the two following properties:
  \begin{itemize}
    \item for $i \leq m$, the canonical image of $X_{i}$ in
    $\K[\mathbf{X}]/I$ is an algebraic integer over
    $\K[X_{m+1},\dots,X_{n}]$;
    \item $\K[X_{m+1},\dots,X_{n}] \cap I= 0$.
  \end{itemize}
  The system $F$ is said to be \emph{in simultaneous Noether position}
  (or in SNP) if for any $1 \leq i \leq m$, the system
  $(f_{1},\dots,f_{i})$ is in Noether position w.r.t the variables
  $X_{1},\dots,X_{i}$.
\end{definition}

The following proposition enumerates useful characterizations of the
Noether position.  They are mostly folklore, but we give a proof for
completeness.
\begin{prop}
  \label{prop:carac-NP}
  Let $m \leq n$, $W$ be a system of weights and $D$ be a system of
  $W$-degrees.  Let $F=(f_{1},\dots,f_{m})$ be a sequence of
  $W$-homogeneous polynomials, with $W$-degree $D$.  The following
  statements are equivalent:
  \begin{eqlist}[NP]
    \item\label{item:NP1} the sequence $F$ is in Noether position
    w.r.t.\ the variables $X_{1},\dots,X_{m}$;
    \item\label{item:NP2} the sequence
    $\Fext \coloneq (f_{1},\dots,f_{m},X_{m+1},\dots,X_{n})$ is
    regular;
    \item\label{item:NP3} the sequence
    $F' \coloneq F(X_{1},\dots,X_{m},0,\dots,0)$ is in Noether
    position w.r.t.\ the variables $X_{1},\dots,X_{m}$;
    \item\label{item:NP4} the sequence $F'$ is regular.
  \end{eqlist}
\end{prop}
\begin{proof}
  (\ref{item:NP1}~$\implies$~\ref{item:NP2}).  \footnote{The proof of~\ref{item:NP1}~$\iff$~\ref{item:NP2} can be found in~\cite{FSV13}, we give it again here for completeness.\ifdraft{{\color{gray}[Footnote pour la version finale]}}{}}
  Let $I$ be the ideal generated by $F$.
  The geometric characterization of Noether position (see e.g.~\cite{MilneAG}) shows that the canonical projection onto the $m$ first coordinates
  \begin{equation}
    \pi : V(I) \to V(\langle X_{1}, \dots ,X_{m} \rangle)\label{eq:6}
  \end{equation}
  is a surjective morphism with finite fibers.  This implies that the
  variety
  \hbox{$V(\langle \Fext \rangle) = \pi^{-1}(0)$} is
  zero-dimensional, and so the sequence is regular.

(\ref{item:NP2}~$\implies$~\ref{item:NP1}).
Let $i \leq m$, we want to show that $X_i$ is integral over the ring $\K[X_{m+1},\dots,X_{n}]$.
Since $\Fext$ defines a zero-dimensional ideal, there exists $n_i \in \N$ such that $X_{i}^{n_{i}} = \LT(f)$ with $f \in \langle \Fext \rangle$ for the \grevlex{} ordering with $X_1 > \dots > X_n$.
By definition of the \grevlex{} ordering, we can assume that $f$ simply belongs to $I$.
This shows that every $X_{i}$ is integral over $\K[X_{i+1},\dots,X_{n}]/I$.
We get the requested result by induction on $i$: first, this is clear if $i=m$.
Now assume that we know that $\K[X_{i},\dots,X_{n}]/I$ is an integral extension of $\K[X_{m+1},\dots,X_{n}]$.
From the above, we also know that $X_{i-1}$ is integral over $\K[X_{i},\dots,X_{n}]$, and so, since the composition of integral homomorphisms is integral, we get the requested result.

Finally, we want to check the second part of the definition of Noether position.
Assume that there is a non-zero polynomial in $\K[X_{m+1},\dots,X_{n}] \cap I$.
Since the ideal is weighted homogeneous, we can assume this polynomial to be weighted homogeneous.
Either this polynomial has degree 0, or it is a non-trivial syzygy between $X_{m+1},\dots,X_{n}$.
So in any case, it contradicts the regularity hypothesis.
  
  (\ref{item:NP2}~$\implies$~\ref{item:NP4}).  For any
  $i \in \{1,\dots,m\}$, write
  $f'_{i} = f_{i}(X_{1},\dots,X_{m},0,\dots,0)$.  Since any
  permutation of a regular sequence is a regular sequence,
  $(X_{m+1},\dots,X_{n},f_{1},\dots,f_{m})$ is a regular sequence,
  that is, for any $1 \leq i \leq m$, $f_{i}$ is not a zero divisor in
  \begin{equation}
    \label{eq:4}
    \KK[X_{1},\dots,X_{n}]/\langle X_{m+1},\dots,X_{n},f_{1},\dots,f_{i-1} \rangle
  \end{equation}
  As a consequence, factoring in the quotient by
  $\langle X_{m+1},\dots,X_{m}\rangle$, $f'_{i}$ is no zero-divisor in
  \begin{equation}
    \label{eq:62}
    \KK[X_{1},\dots,X_{m}]/\langle f'_{1},\dots,f'_{i-1} \rangle.
  \end{equation}

(\ref{item:NP4}~$\implies$~\ref{item:NP2}).
For any $i$, write $f_{i} = f'_{i} + r_{i}$ with $f'_{i} \in \KK[X_{1},\dots,X_{m}]$, and $r_{i} \in \langle X_{m+1}, \dots, X_{n}\rangle$.
Let $1 \leq i \leq n$.
Assume that $gf_{i} \in \langle X_{m+1},\dots,X_{n},f_{1},\dots,f_{i-1} \rangle$:
  \begin{align}
    \label{eq:63}
    gf_{i} = gf'_{i} + gr_{i} &= \sum_{j=1}^{i-1} g_{j}f_{j} + \sum_{j=m+1}^{n} g_{j}X_{j} \\
    \label{eq:80}
    &= \sum_{j=1}^{i-1} g_{j}f'_{j} + R & \text{with
      $R \in \langle X_{m+1},\dots,X_{n}\rangle$}.
  \end{align}
  As a consequence, considering only the monomials in
  $\KK[X_{1},\dots,X_{m}]$
  \begin{equation}
    \label{eq:81}
    g'f'_{i} = \sum_{j=1}^{i-1}g_{j}f'_{j} \text{ where } g'=g(X_{1},\dots,X_{m},0,\dots,0).
  \end{equation}
  Since $F'$ is regular,
  $g' \in \langle f'_{1},\dots,f'_{i-1}\rangle$:
  \begin{equation}
    \label{eq:82}
    g = g'+r \in \langle f'_{1},\dots,f'_{i-1} \rangle  + \langle X_{m+1},\dots,X_{m}\rangle = \langle f_{1},\dots,f_{i-1} \rangle + \langle X_{m+1},\dots,X_{m}\rangle.
  \end{equation}
  And indeed, $f_{i}$ is no zero-divisor in
  $\KK[X_{1},\dots,X_{n}]/\langle
  X_{m+1},\dots,X_{n},f_{1},\dots,f_{i-1}\rangle$.
  It means that $(X_{m+1},\dots,X_{n},f_{1},\dots,f_{m})$ is a
  regular sequence.  By permutation, we conclude that
  $(f_{1},\dots,f_{m},X_{m+1},\dots,X_{n})$ is a regular sequence.

  (\ref{item:NP4}~$\iff$~\ref{item:NP3}).  The sequence
  $F'=(f'_{1},\dots,f'_{m}) \in \KK{[X_{1},\dots,X_{m}]}^{m}$ is
  regular if and only if the sequence
  $(f'_{1},\dots,f'_{m},X_{m+1},\dots,X_{n})$ is regular.  The
  equivalence between~\ref{item:NP3} and~\ref{item:NP4} is then a
  mirror of the equivalence between~\ref{item:NP1} and~\ref{item:NP2}.
\end{proof}

\subsection{Reverse chain-divisible systems of weights and their
  properties}
\label{sec:rcd-weights}

Let $W$ be a system of weights.  Several properties from the
homogeneous case turn out to be no longer true in the weighted case.
For example, properties such as the Noether normalization lemma are
no longer available, since in general, we cannot write any non trivial
weighted homogeneous change of coordinates.  However, if we add some
constraints on the system of weights, some of these properties can be
proved in a weighted setting.  More precisely, we will consider
\emph{reverse chain-divisible} systems of weights, defined as follows.
\begin{definition}
  We say that $W$ is \emph{reverse chain-divisible}
  if we have
  \begin{equation}
    \label{eq:44}
    w_{n} \divides w_{n-1} \divides \dots \divides w_{1}
  \end{equation}
  In this situation, the weights are coprime if and only if $w_{n}=1$.
\end{definition}
\begin{rmk}
  The name ``chain-divisible'' can be found in~\cite{alfonsin2005}, referring to a notion introduced in~\cite{alfonsin1998}.
\end{rmk}

In this setting, many results from the homogeneous case can now be
adapted to the weighted homogeneous case.  For example, the Noether
normalization lemma states that for homogeneous polynomials with an
infinite base field, all regular sequences are in Noether position up
to a generic linear change of coordinates.  In the weighted
homogeneous case with reverse chain-divisible weights, all regular
sequences 
are in Noether position, up to
a weighted homogeneous change of coordinates, with $W$-degree $W$.
More precisely, in the weighted homogeneous case, we can prove the
following version of the Noether normalization lemma
(see~\cite[lem.~13.2.c]{eisenbud95} for the homogeneous version of
this lemma):

\begin{lemma}[Noether normalization lemma, weighted case]
  Let $\KK$ be an infinite field, $W$ be a reverse chain-divisible system of weights and $f \in R=\KK[X_{1},\dots,X_{r}]$ be a non-constant polynomial, $W$-homogeneous with $W$-degree $d$.
  Then there are elements $X'_{1},\dots,X'_{r-1}\in R$ such that $R$ is a finitely generated module over $\KK[X'_{1},\dots,X'_{r-1},f]$.
  Furthermore, if the field has characteristic $0$ or large enough, there exists a dense Zariski-open subset $U \subset \KK^{r-1}$ such that for all $(a_{i}) \in U$, one can choose $X'_{i} = X_{i} - a_{i}X_{r}^{w_{i}/w_{r}}$.
\end{lemma}
\begin{proof}
  We follow the proof of~\cite[lem.~13.2.c]{eisenbud95}.  For any
  $1 \leq i \leq r-1$, let $a_{i} \in \KK$, and let
  $X'_{i} = X_{i} - a_{i}X_{r}^{w_{i}/w_{r}}$.  We need to show that
  for generic $a_{i}$, under this change of variables, $f$ is monic in
  $X_{r}$:
  \begin{align}
    f(X_{1},\dots,X_{r}) &= f(X'_{1} + a_{1}X_{r}^{w_{1}/w_{r}},
    X'_{2} + a_{2}X_{r}^{w_{2}/w_{r}}, \dots, X_{r-1} + a_{r-1}X_{r}^{w_{r-1}/w_{r}}) \\
    &= f(a_{1},\dots,a_{r-1},1)X_{r}^{d} + \dots
  \end{align}
  So the set of all $a_{i}$'s such that $f$ is monic in $X_{r}$ is
  exactly the set of all $a_{i}$'s such that
  $f(a_{1},\dots,a_{r-1},1) \neq 0$, and since $f$ is $W$-homogeneous
  non-constant, this is a non-empty open subset of $\KK^{r-1}$.
\end{proof}

Then, as in the homogeneous case \citep[th.~13.3]{eisenbud95}, a
consequence of this lemma is Noether's normalization theorem, which we
restate in a weighted setting:
\begin{theorem}
  \label{thm:gen-linchg}
  Let $W$ be a reverse chain-divisible system of weights, and let $F$ be a
  \hbox{$W$-homogeneous} zero-dimensional regular sequence in
  $\KK[X_{1},\dots,X_{n}]$.  Then, for a generic choice of
  $W$-homogeneous polynomials $P_{i}$ with $W$-degree $w_{i}$, the
  change of variable
  \begin{equation}
    \label{eq:90}
    X_{i} = X'_{i} + P_{i}(X_{i+1},\dots,X_{n}),
  \end{equation}
  is such that $F(X_{1}(\mathbf{X'}),\dots,X_{n}(\mathbf{X'}))$ is in
  simultaneous Noether position with respect to the order
  $X'_{1} > X'_{2} > \dots > X'_{n}$.
\end{theorem}

Another property of reverse chain-divisible weights is the following proposition.
In the homogeneous case, if $d_{1} \leq d_{2}$ are two non-negative integers, then any monomial with degree $d_{2}$ is divisible by a monomial with degree $d_{1}$.
When the system of weights is reverse chain-divisible, the following proposition states a similar result for the weighted case.


\begin{prop}
  Assume that $W=(w_{1},\dots,w_{n})$ is a system of weights, such that \hbox{$w_{1} \geq w_{2} \geq \dots \geq w_{n}$}.
  The following statements are equivalent:
  \begin{enumerate}
    \item \label{item:rcd1} The system of weights $W$ is reverse
    chain-divisible;
    \item \label{item:rcd2} Let $d_{1} \leq d_{2}$ positive integers,
    $i \in \{1,\dots,n\}$, and $m_{2}$ a monomial of $W$-degree
    $d_{2}$.  Assume that $w_{i}$ divides $d_{1}$, and that $m_{2}$ is
    not divisible by any of the variables $X_{1},\dots,X_{i-1}$.  Then
    there exists a monomial $m_{1}$ with $W$-degree $d_{1}$, such that
    \hbox{$m_{1}\divides m_{2}$}.
  \end{enumerate}
\end{prop}
\begin{proof}
  (\ref{item:rcd1}~$\implies$~\ref{item:rcd2}).
  Fix $d_{1}$.
  We shall prove by induction over $d_{2}$ that for any monomial $m_{2}$ with $W$-degree $d_{2}$ satisfying the hypotheses of~(\ref{item:rcd2}), there exists a monomial $m_{1}$ with $W$-degree $d_{1}$ dividing $m_{2}$.
  The case $d_{2} = d_{1}$ is trivial.

  Assume that $d_{2} > d_{1}$, and let $m_{2}$ be a monomial of
  $W$-degree $d_{2}$.  Let $j$ be the greatest index of a variable
  dividing $m_{2}$, write $m_{2} = X_{j}^{\alpha} m'_{2}$, where
  $m'_{2}$ is a monomial in $\K[X_{i},\dots,X_{j-1}]$, with $W$-degree
  $d'_{2} = d_{2}-w_{j}\alpha$.  If $d'_{2} \geq d_{1}$, the result
  follows by induction.  If $i=j$, then
  $m_{2}= X_{i}^{\alpha}$, and $m_{1} \coloneq X_{i}^{d_{1}/w_{i}}$ has
  $W$-degree $d_{1}$ and divides $m_{2}$.  So we can assume that
  $d'_{2} < d_{1}$ and that $i < j$.

Since $W$ is a reverse chain-divisible system of weights, $w_{j-1}$ divides $w_{k}$ for any $k$ in $\{1, \dots, j-1\}$.
Hence, since $m_{2} \in \KK[X_{i},\dots,X_{j}]$ and $m'_{2} \in \KK[X_{i},\dots,X_{j-1}]$, $d_{2} \equiv 0 \mod{w_{j}}$ and $d'_{2} \equiv 0 \mod{w_{j-1}}$.
By hypothesis, $d_{1}$ is divisible by $w_{i}$, and in particular it is divisible by $w_{j-1}$.
All in all, this shows that $d_{1}-d'_{2}$ is divisible by $w_{j-1}$, and so it is divisible by $w_{j}$.
Let
  \begin{equation}
    m_{1} = m'_{2}\cdot X_{j}^{(d_{1}-d'_{2})/w_{j}}.
    \label{eq:niceW-mon}
  \end{equation}
  The monomial $m_{1}$ has $W$-degree $d_{1}$ and divides $m_{2}$.

  (\ref{item:rcd2}~$\implies$~\ref{item:rcd1}).  Assume that $W$ is a system
  of weights which is not reverse chain-divisible, we shall find
  integers $d_{1} \leq d_{2}$ and a monomial $m_{2}$ with $W$-degree
  $d_{2}$ which is not divisible by any monomial of $W$-degree
  $d_{1}$.

Since $W$ is not reverse chain-divisible, there exists $i$ such that $w_{i+1}$ does not divide $w_{i}$.
In particular, $\gcd(w_{i},w_{i+1}) < w_{i}$ and $\gcd(w_{i},w_{i+1}) < w_{i+1}$.
Without loss of generality, we may consider only the variables $X_{i},X_{i+1}$.
Let $d_{1} = w_{i}w_{i+1}$, \hbox{$d_{2} = d_{1} + \gcd(w_{i},w_{i+1})$}.
By Paoli's lemma (see for example~\cite[chap.~264]{Lucas1891} or the discussion after~\cite[th.~5.1]{Niven1991}), there exists exactly
  \begin{equation}
    \label{eq:84}
    \left\lfloor \frac{d_{2}}{w_{i}w_{i+1}} \right\rfloor
    = \left\lfloor 1 + \frac{\gcd(w_{i},w_{i+1})}{w_{i}w_{i+1}}  \right\rfloor
    = 1
  \end{equation}
  couple of non-negative integers $a,b$ such that $aw_{i} + bw_{i+1}=d_{2}$.
  Let $m_{2}$ be the monomial $X_{i}^{a}X_{i+1}^{b}$.
  The $W$-degree $d_{1}$ is divisible by $w_{i}$, and $m_{2}$ is not divisible by $X_{1},\dots,X_{i-1}$.
  The maximal divisors of $m_{2}$ are
  \begin{gather}
    \label{eq:niceW-mon1}
    \frac{m_{2}}{X_{i}} = X_{i}^{a-1}X_{i+1}^{b} \text{ with
      $W$-degree
      $d_{2}-w_{i} = d_{1} + \gcd(w_{i},w_{i+1}) - w_{i} < d_{1}$;} \\
    \label{eq:niceW-mon2}
    \frac{m_{2}}{X_{i+1}} = X_{i}^{a}X_{i+1}^{b-1} \text{ with
      $W$-degree
      $d_{2}-w_{i+1} = d_{1} + \gcd(w_{i},w_{i+1}) - w_{i+1} <
      d_{1}$.}
  \end{gather}
  As a consequence, $m_{2}$ is not divisible by any monomial of
  $W$-degree $d_{1}$.
\end{proof}


This proposition essentially states that the staircase of a
$W$-homogeneous ideal is reasonably shaped when $W$ is a reverse
chain-divisible system of weights. 
For example, let $W$ be a reverse chain-divisible system of weights,
and let $I$ be the ideal generated by all monomials of $W$-degree
$w_{1}$ (that is, the least common multiple of the weights).  Then
the proposition proves that $I$ contains all monomials of $W$-degree
greater than $w_{1}$.

If on the other hand the system of weights is not reverse
chain-divisible, this property needs not hold.  For example, consider
the algebra $\K[X_{1},X_{2},X_{3}]$ graded w.r.t. the system of
weights $W=(3,2,1)$, the least common multiple of the weights being
$6$, and let $I$ be the ideal generated by all monomials of $W$-degree
$6$.  Consider the monomial $X_{1}X_{2}^{2}$: it has $W$-degree $7$,
yet it is not divisible by any monomial with $W$-degree $6$, and so it
does not belong to the ideal $I$.

\subsection{Genericity}
\label{sec:Genericity}

We shall give some results about the genericity of regularity and
Noether position for weighted homogeneous sequences.  The fact that
they define Zariski-open subsets of the sets of sequences of a given
weighted degree is classical.  For regular sequences, see for
example~\cite[sec.~2]{Pardue2010}.  The proof for sequences in
(simultaneous) Noether position is a simple extension of the statement
for regular sequences.  However, we provide here a sketch of these
proofs for completeness.

\begin{prop}
  \label{thm:gen-1}
  Let $m \leq n$ be two integers, $W=(w_{1},\dots,w_{n})$ a system of
  weights, and $D=(d_{1},\dots,d_{m})$ a system of $W$-degrees.
  Then
  \begin{itemize}
    \item the set of regular sequences,
    \item the set of sequences in Noether position with respect to the
    variables $X_{1},\dots,X_{m}\,$, and
    \item the set of sequences in simultaneous Noether position
    w.r.t. the order $X_{1} > \dots > X_{m}$
  \end{itemize}
  are Zariski-open subsets of the affine space of $W$-homogeneous
  polynomials with $W$-degree $D$.
\end{prop}

\begin{proof}
  We shall prove that regular sequences form a Zariski-open subset of
  the affine space of $W$-homogeneous polynomials of $W$-degree $D$.
  The openness of Noether position will then be a corollary, since by
  Proposition~\ref{prop:carac-NP}, $(f_{1},\dots,f_{m})$ is in Noether
  position w.r.t the variables $X_{1},\dots,X_{m}$ if and only if
  $(f_{1},\dots,f_{m},X_{m+1},\dots,X_{n})$ is regular.  As for 
  sequences in simultaneous Noether position, they will be given by
  the intersection of $m$ open subsets, stating that the sequences
  $(f_{1},\dots,f_{i})$, $i \in \{1,\dots,m\}$, are in Noether
  position w.r.t.  the variables $X_{1},\dots,X_{i}$.

  Let $F=(f_{1},\dots,f_{m})$ be a family of $m$ generic quasi-homoge\-neous polynomials, that is polynomials in $\KK[\mathbf{a}][\mathbf{X}]$, whose coefficients are algebraically independent parameters $a_{k}$.
  We want to prove that regular sequences are characterized by some polynomial in these coefficients $a_{k}$ being non-zero, which implies that they belong to a Zariski-open set.
  Write $I = \langle F \rangle$.
Since the Hilbert series~\eqref{eq:HSreg} characterizes regular sequences, $F$ is regular if and only if $I$ contains all monomials of $W$-degree between $\HI(I)+1$ and $\HI(I)+ \max\{w_{i}\}$, where $\HI(I)$ is given by $\sum (d_{i} - w_{i})$.
This expresses that a given set of linear equations has solutions, and so it can be coded as some determinants being non-zero, as polynomials in the coefficients $a_{k}$.
\end{proof}

This states that the set of regular sequences, sequences in Noether
position and sequences in simultaneous Noether position are
Zariski-dense subsets if and only if they are not empty.
Unfortunately, depending on the weights and the weighted degrees,
there may exist no regular sequence, and thus no sequences in
(simultaneous) Noether position either.  For example, let $W=(2,5)$
and $D=(4,8)$, the only $W$-homogeneous sequence with $W$-degree $D$
in $\KK[X,Y]$ is (up to scalar multiplication) $(X^{2}, X^{4})$, and
it is not regular.  However, this is only the case for very specific
systems of $W$-degrees, for which there does not exist enough
monomials to build non-trivial sequences.

\begin{definition}
  Let $m \leq n$ be two integers, $W=(w_{1},\dots,w_{n})$ a system of
  weights, and $D=(d_{1},\dots,d_{m})$ a system of $W$-degrees.  We
  say that $D$ is \emph{$W$-compatible} if there exists a regular
  $W$-homogeneous sequence in $\KK[X_{1},\dots,X_{n}]$ with $W$-degree
  $D$.  We say that $D$ is \emph{strongly $W$-compatible} if for any
  $1 \leq i \leq m$, $d_{i}$ is divisible by $w_{i}$.%
\end{definition}

Using these definitions, we can identify the cases where the
properties of being regular, in Noether position or in simultaneous
Noether position are generic.

\begin{prop}
  \label{thm:gen-2}
  Let $m \leq n$ be two integers, $W=(w_{1},\dots,w_{n})$ a system of
  weights, and $D=(d_{1},\dots,d_{m})$ a system of $W$-degrees.  For any
  $1 \leq i \leq m$, write $W_{i} \coloneq (w_{1},\dots,w_{i})$ and
  $D_{i} \coloneq (d_{1},\dots,d_{i})$.  Write $A_{W,D}$ the affine
  space of $W$-homogeneous sequences of $W$-degree $D$.  Then the
  following statements are true:
  \begin{enumerate}
    \item\label{item:gen-1} if $D$ is $W$-compatible, then regular
    sequences form a Zariski-dense subset of $A_{W,D}$;
    \item\label{item:gen-2} if $D$ is $W_{m}$-compatible, then
    sequences in Noether position with respect to the variables
    $X_{1},\dots,X_{m}$ form a Zariski-dense subset of $A_{W,D}$;
    \item\label{item:gen-4} if $D$ is strongly $W$-compatible, then
    $D$ is $W$-compatible, $W_{m}$-compatible, and for any $i$,
    $D_{i}$ is $W_{i}$-compatible;
    \item\label{item:gen-5} if $m=n$, $D$ is $W$-compatible and $W$ is
    reverse chain-divisible, then, up to some reordering of the
    degrees, $D$ is strongly $W$-compatible.
  \end{enumerate}
\end{prop}
\begin{proof}
  The proofs of statements~\ref{item:gen-1} and~\ref{item:gen-2}
  follow the same technique: by Theorem~\ref{thm:gen-1}, we know that
  the sets we consider are Zariski-open in $A_{W,D}$.  So in order to
  prove the density, we only need to prove that they are non empty.
  Statement~\ref{item:gen-1} is exactly the definition of the
  $W$-compatibility.
  
  For statement~\ref{item:gen-2}, by $W_{m}$-compatibility, we know
  that there exists a $W$-homogeneous sequence $F=(f_{1},\dots,f_{m})$
  with $W$-degree $D$ in $\KK[X_{1},\dots,X_{m}]$, which is regular.
  As a consequence, the sequence
  $(f_{1},\dots,f_{m},X_{m+1},\dots,X_{n})$ is regular, and from the
  characterization~\ref{item:NP4} of Noether position
  (prop.~\ref{prop:carac-NP}), this means that $F$ is in Noether
  position with respect to the variables $X_{1},\dots,X_{m}$.
    
  In order to prove statement~\ref{item:gen-4}, we need to exhibit
  regular sequences of length $i$ in $\KK[X_{1},\dots,X_{i}]$ for any
  $1 \leq i \leq m$.  For $1 \leq i \leq m$, write
  $F_{i} = (X_{1}^{d_{1}/w_{1}},\dots,X_{i}^{d_{i}/w_{i}})$, it is
  regular and each polynomial lies in $\KK[X_{1},\dots,X_{i}]$.

  Finally, statement~\ref{item:gen-5} is a consequence of
  Theorem~\ref{thm:gen-linchg}.  Let $W$ be a reverse chain-divisible
  system of weights, and $D$ a $W$-compatible system of $W$-degrees.
  Up to reordering, we can assume that the polynomials are ordered so that
  $d_{1} \geq d_{2} \geq \dots \geq d_{n}$; this does not cancel the
  $W$-compatibility.  Let $F=(f_{1},\dots,f_{n})$ be a regular
  sequence, $W$-homogeneous with $W$-degree $D$.  By
  Theorem~\ref{thm:gen-linchg}, there exist polynomials
  $P_{i}(X_{i+1},\dots,X_{n})$ which are $W$-homogeneous with
  $W$-degree $w_{i}$, and such that $F$, under the change of variables
  $X_{i}=X'_{i}+P_{i}(X_{i+1},\dots,X_{n})$, is in simultaneous
  Noether position with respect to the order
  $X'_{1} > X'_{2} > \dots > X'_{n}$.
  From the characterization~\ref{item:NP4} of Noether position, that
  means in particular that for any $i \in \{1,\dots,n\}$,
  $f_{i}(X_{1}(X'_{1},\dots,X'_{i}),\dots,X_{n}(X'_{1},\dots,X'_{i}))$
  belongs to a regular sequence, and thus is not zero.  And by
  definition of reverse chain-divisible weights, its $W$-degree
  $d_{i}$ is a sum of multiples of $w_{i}$, and so it is itself a
  multiple of $w_{i}$.
\end{proof}

\begin{remark}
  The statement~\ref{item:gen-5} is a converse of~\ref{item:gen-4} in
  the reverse chain-divisible case.  In the non-reverse
  chain-divisible case, that converse is false: let $W=(3,2)$,
  $D=(6,5)$ and consider $F=(X^{2}+Y^{3},XY)$ in $\KK[X,Y]$.  The
  sequence $F$ is in simultaneous Noether position w.r.t.\ the order
  $X > Y$, yet $5$ is neither divisible by $3$ nor by $2$.

  The weaker converse that if $D$ is $W$-compatible, then $D$ is
  $W_{m}$-compatible is also false: with the same weights and algebra,
  let $D=(5)$, the only polynomial with $W$-degree $5$ is (up to
  scalar multiplication) $f=XY$.  It is non-zero, so $(f)$ is a
  regular sequence, but $(f,Y)$ is not regular, so $(f)$ is not in
  Noether position w.r.t $X$.
\end{remark}

\begin{remark}
  These examples lead to the following attempt at writing a general
  characterization of
  $W$-compatibility.

  Let $n$ be a positive integer, $W=(w_{1},\dots,w_{n})$ a system of
  weights, and $D=(d_{1},\dots,d_{n})$ a system of $W$-degrees.
  Further assume that
  \begin{itemize}
    \item for all $i \in \{1,\dots,n\}$,
    $\KK{[\mathbf{X}]}_{d_{i}} \neq 0$
    \item the formal series
    \begin{equation}
      \label{eq:88}
      S_{D,W}(T) = \frac{\prod_{i=1}^{n}(1-T^{d_{i}})}{\prod_{i=1}^{n}(1-T^{w_{i}})}
    \end{equation}
    is a polynomial.
  \end{itemize}
  Is $D$ necessarily $W$-compatible?

The answer is \emph{no}: take the system of weights $W=(3,5,11)$, and
the system of $W$-degrees $D=(165,19,19)$.  Note that $165$ is the
product of the weights, and $19$ the sum of the weights.  The series
\begin{equation}
  S_{D,W}(T) = \frac{(1-T^{165})\cdot(1-T^{19})\cdot(1-T^{19})}{(1-T^{3})\cdot(1-T^{5})\cdot(1-T^{11})} = 1 + T^{3} + \dots + T^{184}
  \label{eq:46}
\end{equation}
is a polynomial.
But at $W$-degree $19$, there are only $2$ monomials, namely $X_{1}X_{2}X_{3}$ and $X_{1}^{3}X_{2}^{2}$, and they are not coprime, so we cannot form a regular sequence of $W$-degrees $(165,19,19)$.
\end{remark}

\section{Regular systems}
\label{sec:Regular-systems}

\subsection{Shape of the Hilbert series of a weighted homogeneous
  complete intersection}
\label{sec:Shape-Hilbert-series}

Let $W=(w_{1},\dots,w_{n})$ be a reverse chain-divisible system of
weights such that $w_{n}=1$, and let $D=(d_{1},\dots,d_{n})$ be a system of $W$-degrees,
such that for any $i \in \{1,\dots,n\}$, $d_{i}$ is divisible by all
of the $w_{j}$'s.  Let $R = \K[X_{1},\dots,X_{n}]$ be a polynomial
algebra graded with respect to $W$.

We use the following notations, as found in~\citep{MS96}:
\begin{itemize}
  \item $\delta_{j} = \sum_{i=1}^{j}(d_{i}-w_{i})$;
  \item $\delta = \delta_{n}$, $\delta^{\ast} = \delta_{n-1}$;
  \item
  $\sigma = \min\left( \delta^{\ast}, \left\lfloor \frac{\delta}{2}
   \right\rfloor \right)$,
  $\sigma^{\ast} = \min\left( \delta_{n-2}, \left\lfloor
    \frac{\delta^{\ast}}{2} \right\rfloor \right)$;
  \item $\mu = \delta - 2\sigma$,
  $\mu^{\ast} = \delta^{\ast} - 2\sigma^{\ast}$.
\end{itemize}
Given a formal series $S(T)=\sum_{d=0}^{\infty} a_{d}T^{d}$, we also define
\begin{equation}
  \begin{aligned}
    \Delta S(T) &= \sum_{d=0}^{\infty} (a_{d}-a_{d-1})T^{d} &&
    \text{(with the convention $a_{-1}=0$)} \\
    &= (1-T) \cdot S(T)
  \end{aligned}
  \label{eq:76}
\end{equation}
and
\begin{equation}
  \label{eq:71}
  \int S = \sum_{d=0}^{\infty} (a_{0} + \dots + a_{d}) T^{d} = \frac{S(T)}{1-T}.
\end{equation}

\begin{lemma}
  \label{lem:lemMS}
  Under the above notations and assumptions, the following properties
  hold.
  \begin{gather}
    \left\{
     \begin{aligned}
       \delta^{\ast} & > \left\lfloor \frac{\delta}{2} \right\rfloor &
       & \iff &  &         & d_{n} - \delta^{\ast} & \leq 0 \\
       \delta^{\ast} & = \left\lfloor \frac{\delta}{2} \right\rfloor &
       & \iff &  & 1  \leq & d_{n} - \delta^{\ast} & \leq 2 \\
       \delta^{\ast} & < \left\lfloor \frac{\delta}{2} \right\rfloor &
       & \iff & & 3 \leq & d_{n} - \delta^{\ast} &
     \end{aligned}
    \right.
    \label{eq:lemMS1} \\
    \sigma = \left\lfloor \frac{\delta}{2} \right\rfloor \implies \mu
    = \delta \mod 2 \in \{0,1\}
    \label{eq:lemMS1-2} \\
    \label{eq:lemMS2}
    0 \leq \mu < d_{n} \\
    \label{eq:lemMS3}
    d_{n-1} \leq d_{n} \implies \sigma^{\ast} + \mu^{\ast} \leq \sigma
  \end{gather}
\end{lemma}
\begin{proof}
  The proof of statements~\eqref{eq:lemMS1} and~\eqref{eq:lemMS2} can
  be found in~\cite[Lemma~2.1]{MS96}.  This proof depends only
  on the value of $w_{n}$, and since we assume it to be $1$, it is
  also valid in our setting.  It also proves~\eqref{eq:lemMS1-2} as a
  side-result.

  For the statement~\eqref{eq:lemMS3}, we proceed by case disjunction
  on the values of $\sigma$.
  \begin{itemize}
    \item If $\sigma = \delta^{\ast}$:
    \begin{align}
      \label{eq:38}
      \sigma^{\ast} + \mu^{\ast} = \delta^{\ast} - \sigma^{\ast} \leq
      \delta^{\ast} = \sigma.
    \end{align}
    \item If $\sigma = \lfloor \delta/2 \rfloor$, then
    $\sigma = \left\lfloor (\delta^{\ast} + d_{n} -1)/2 \right\rfloor$
    which implies $2\sigma = \delta^{\ast} + d_{n} - 1 - \mu$ and
    $\mu= \delta \mod 2 \in \{0,1\}$ (from
    statement~\eqref{eq:lemMS1-2}).  Now consider the possible values
    of $\sigma^{\ast}$:
    \begin{itemize}
      \item if $\sigma^{\ast} = \lfloor \delta^{\ast}/2 \rfloor$, then
      $\mu^{\ast}=\delta^{\ast} \mod 2$, and thus
      $2\sigma = 2\sigma^{\ast} + \mu^{\ast} + d_{n} -1 - \mu$.  It
      implies that $d_{n} -1 - \mu + \mu^{\ast}$ is even, we shall
      prove that it is greater than or equal to $0$.

      From statement~\eqref{eq:lemMS2}, $d_{n}-1 - \mu \geq 0$, so if
      $\mu^{\ast}=0$, we are done.  If $\mu^{\ast} = 1$, by parity
      $d_{n}-1-\mu$ is odd, and thus
      $d_{n}-1-\mu \geq 1 = \mu^{\ast}$.

      It implies that:
      \begin{equation}
        \label{eq:51}
        2\sigma = 2\sigma^{\ast} + \mu^{\ast} + d_{n} - 1 - \mu \geq 2\sigma^{\ast} + 2\mu^{\ast};
      \end{equation}
      \item otherwise, $\sigma^{\ast} = \delta^{\ast\ast}$, and in
      that case
      \begin{equation}
        \sigma^{\ast}+\mu^{\ast} = \delta^{\ast} - \sigma^{\ast}  = \delta^{\ast} - \delta^{\ast\ast} = d_{n-1}-w_{n-1}
        \label{eq:7}
      \end{equation}
      which implies that:
      \begin{equation}
        \label{eq:53}
        d_{n}-1 \geq \sigma^{\ast} + \mu^{\ast} \text{ (since $w_{n-1} \geq w_{n}$ and $d_{n-1} \leq d_{n}$)}
      \end{equation}
      and
      \begin{equation}
        \label{eq:54}
        \delta^{\ast} = \delta^{\ast\ast} + d_{n-1}-w_{n-1} \geq \sigma^{\ast} + \mu^{\ast}.
      \end{equation}
      So we have:
      \begin{align}
        \label{eq:49}
        2\sigma &= \delta^{\ast} + d_{n} - 1 - \mu \\
        \label{eq:52}
        &\geq \sigma^{\ast} + \mu^{\ast} + \sigma^{\ast} + \mu^{\ast}
        - \mu.
      \end{align}
      Recall that $\mu \in \{0,1\}$, so by parity,
      $2\sigma \geq 2\sigma^{\ast} + 2\mu^{\ast}$, hence
      $\sigma \geq \sigma^{\ast} + \mu^{\ast}$.\qedhere
    \end{itemize}
  \end{itemize}
\end{proof}

The following theorem is a description of the shape of the Hilbert
series of a zero-dimensional complete intersection.  It states that it
is a self-reciprocal (or palindromic) polynomial, that is a polynomial
with symmetrical coefficients, and that these coefficients increase at
small degrees, then station, then decrease again.  Furthermore,
between every strict increase, they reach a step, which has width
$w_{n-1}$. For an example, see figure~\ref{fig:shape-HS-reg}, where
the width of the steps is 3, and the width of the central plateau is
5.

This is a generalization of a known result in the
homogeneous case, which has been proved for example
in~\cite[prop.~2.2]{MS96} (we will follow that proof for the weighted
case).  In the homogeneous case, there is no such step in the growth
of the coefficients, and they are strictly increasing, then
stationary, then strictly decreasing. 

\begin{theorem}
  \label{thm:HS-regseq}
  Let $W=(w_{1},\dots,w_{n})$ be a reverse chain-divisible system of
  weights, and $D=(d_{1},\dots,d_{n})$ a system of degrees such that
  for any $i \in \{1,\dots,n\}$, $d_{i}$ is divisible by $w_{1}$.
  Consider the formal series
  \begin{equation}
    \label{eq:26}
    S_{W,D}(T) = \frac{\prod_{i=1}^{n}(1-T^{d_{i}})}{\prod_{i=1}^{n}(1-T^{w_{i}})}
    = \sum_{d=0}^{\delta}a_{d}T^{d}
  \end{equation}
  The series $S_{W,D}$ is a self-reciprocal polynomial in $T$ (i.e.\
  for any $d \leq \delta$, $a_{d} = a_{\delta - d}$) and its
  coefficients satisfy the inequalities:
  \begin{equation}
    \begin{aligned}
      \label{eq:sym}
      \forall d \in \{0, \dots, \sigma-1\}, && a_{d} \leq a_{d+1} \\
      \forall d \in \{\sigma, \dots, \sigma + \mu-1\}, && a_{d} = a_{d+1} \\
      \forall d \in \{\sigma + \mu, \dots, \delta\}, && a_{d} \geq
      a_{d+1}
    \end{aligned}
  \end{equation}
  Furthermore, if $d<\sigma$ (resp. $d>\sigma+\mu$), the coefficients
  increase (resp.\ decrease) with steps, and these steps have width
  $w_{n-1}$:
  \begin{equation}
    \label{eq:55}
    \forall d \in \{0, \dots, \sigma-1\},\, a_{d}-a_{d-1}
    \begin{cases}
      >0 & \text{if $w_{n-1}$ divides $d$} \\
      =0 & \text{otherwise.}
    \end{cases}
  \end{equation}
\end{theorem}

\begin{figure}
  \centering
  \begin{tikzpicture}[x=0.7cm,y=0.7cm]
  \tikzset{dot/.style={draw, circle, inner sep=1pt, fill, color=black}}
  
  \draw [color=black!10,step=1] (-0.1,-0.1) grid (18.8,3.8);

  \drawSeries{{
    0 / 1, 1 / 1, 2 / 1,%
    3 / 2, 4 / 2, 5 / 2,%
    6 / 3, 7 / 3, 8 / 3, 9 / 3, 10 / 3, 11 / 3,%
    12 / 2, 13 / 2, 14 / 2,%
    15 / 1, 16 / 1, 17 / 1%
  }}{p}{0}{0}{0}

  \draw (p17) |- (18,0) node [dot] {};
  \draw [axis] (0,0) -- (19,0) node [anchor=north] {$d$};
  \draw [axis] (0,0) -- (0,4)  node [anchor=east]  {$a_{d}$};;

  \node [anchor=north] (s)   at (6,0) {$\sigma = 6$};
  \node [anchor=north] (spm) at (11,0) {$\sigma+\mu = 11$};
  \node [anchor=north] (d)   at (17,0) {$\delta = 17$};

  \draw [dashed] (s)   -- (p6);
  \draw [dashed] (spm) -- (p11);
  \draw [dashed] (d)   -- (p17);

  \draw [<->,thick] (6,1.5)  -- (11,1.5) node [pos=0.5,anchor=north] {$\mu=5$};
  \draw [<->,thick] (14,2.5) -- (17,2.5) node [pos=0.5,anchor=south] {$w_{n-1}=3$};
  
\end{tikzpicture}


  \caption{Shape of the Hilbert series of a $W$-homogeneous complete
    intersection for \hbox{$W=(3,3,1)$} and \hbox{$D=(9,6,3)$}}
  \label{fig:shape-HS-reg}
\end{figure}
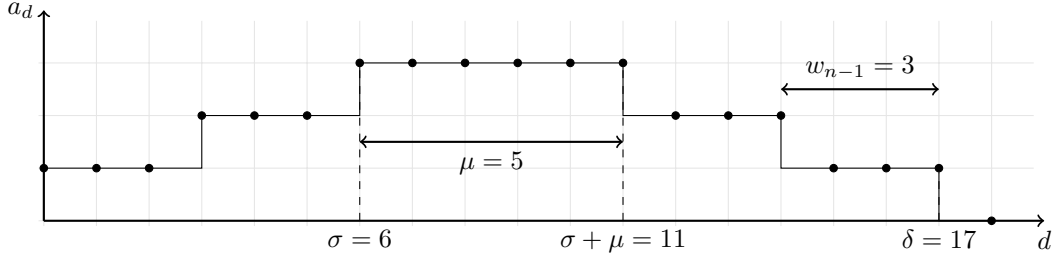

\begin{proof}
  We adapt the proof from~\cite[Prop.~2.2]{MS96} for the homogeneous case to the weighted case.
  Up to permutation of the $d_{i}$'s, we can assume that for any $i$, $d_{i} \geq d_{i-1}$.
  We proceed by induction on $n$.
  The result for the case $n=1$ is a consequence of the homogeneous case, since $w_{n}=1$.

  Let $n > 1$.  Let
  $\bar{W}^{\ast} = (w_{1}/w_{n-1},\dots,w_{n-1}/w_{n-1})$ and
  $\bar{D}^{\ast} = (d_{1}/w_{n-1},\dots,d_{n-1}/w_{n-1})$, and
  consider the series
  \begin{equation}
    \bar{S}^{\ast} \coloneq S_{W^{\ast},D^{\ast}}
    = \frac{\prod_{i=1}^{n-1}(1-T^{d_{i}/w_{n-1}})}{\prod_{i=1}^{n-1}(1-T^{w_{i}/w_{n-1}})}
    = \sum_{d=0}^{\delta}\bar{a}^{\ast}_{d}T^{d}.
    \label{eq:64}
  \end{equation}

  The Hilbert series $S$ can be computed from $\bar{S}^{\ast}$ with
  \begin{equation}
    \label{eq:40}
    S(T) = \frac{1-T^{d_{n}}}{1-T}\bar{S}^{\ast}(T^{w_{n-1}}) = (1-T^{d_{n}})\cdot \int \bar{S}^{\ast}(T^{w_{n-1}}),
  \end{equation}
  and so for any $d$, we have:
  \begin{gather}
    \label{eq:41}
    a_{d} = a^{\ast}_{d-d_{n}+1} + \dots + a^{\ast}_{d} \\
    \label{eq:42}
    a'_{d} \coloneq a_{d} - a_{d-1} = a^{\ast}_{d} -
    a^{\ast}_{d-d_{n}}
  \end{gather}
  where
  \begin{equation}
    \label{eq:43}
    a^{\ast}_{d} =
    \begin{cases}
      \bar{a}^{\ast}_{\bar{d}} & \text{if $d = \bar{d}w_{n-1}$} \\
      0 & \text{otherwise.}
    \end{cases}
  \end{equation}
  This proves that the polynomial is self-reciprocal:
  \begin{align}
    \label{eq:85}
    a_{\delta-d} &= a^{\ast}_{\delta - d-d_{n}+1} + \dots + a^{\ast}_{\delta - d} \\
    \label{eq:86}
    &= a^{\ast}_{d-d_{n}+1} + \dots + a^{\ast}_{d} \text{ since, by induction hypothesis, $\bar{S}^{\ast}$ is self-reciprocal} \\
    \label{eq:87}
    &= a_{d}
  \end{align}

  To prove the properties regarding the sign of
  $a'_{d} = a_{d}-a_{d-1}$, we shall consider two cases, according to
  the value of $d_{n}$.
  \begin{itemize}
    \item If $d_{n} \geq \delta^{\ast}+1$, then from
    statement~\eqref{eq:lemMS1} in Lemma~\ref{lem:lemMS}, and the
    definition of $\sigma$ and $\mu$, $\sigma = \delta^{\ast}$ and
    $\sigma+\mu = d_{n}-1$.  Let $0 \leq d \leq \sigma$, then
    $d \leq \delta^{\ast} < d_{n}$, and thus:
    \begin{equation}
      \label{eq:56}
      a'_{d} = a^{\ast}_{d} =
      \begin{cases}
        \bar{a}^{\ast}_{d/w_{n-1}}>0 & \text{if $w_{n-1}$ divides $d$;} \\
        0 & \text{otherwise.}
      \end{cases}
    \end{equation}
    Let $d \in \{\sigma+1, \dots, \sigma+\mu\}$, that implies that
    $\delta^{\ast} < d \leq d_{n}-1$, and thus:
    \begin{equation}
      \label{eq:57}
      a'_{d} = a^{\ast}_{d} = 0 \text{ (since $\delta^{\ast}$ is the degree of $S^{\ast}$).}
    \end{equation}
    \item If $d_{n} \leq \delta^{\ast}$, then from
    statement~\eqref{eq:lemMS1} again,
    $\sigma = \lfloor \delta/2 \rfloor$ and $\mu = \delta \mod 2$.
    Let $d \leq \sigma$, we want to prove that $a_{d}-a_{d-1}$ is
    greater or equal to zero, depending on whether $d$ is divisible by
    $w_{n-1}$.  We shall consider two ranges of values for $d$:
    \begin{itemize}
      \item if $d \leq \sigma^{\ast}+\mu^{\ast}$, then
      $d-d_{n} \leq \sigma^{\ast} + \mu^{\ast} - d_{n} <
      \sigma^{\ast}$
      (since $\mu^{\ast} < d_{n}$).  Recall that
      $a'_{d} = a^{\ast}_{d} - a^{\ast}_{d-d_{n}}$.  By hypothesis,
      $d_{n}$ is divisible by $w_{n-1}$, and so, either both $d$ and
      $d-d_{n}$ are divisible by $w_{n-1}$, or both are not.  Thus,
      \begin{equation}
        \label{eq:59}
        a'_{d}\,
        \begin{cases}
          > 0 & \text{if both $d$ and $d-d_{n}$ are divisible by $w_{n-1}$} \\
          = 0 & \text{if neither $d$ nor $d-d_{n}$ is divisible by
            $w_{n-1}$;}
        \end{cases}
      \end{equation}
      \item if $\sigma^{\ast}+\mu^{\ast} < d \leq \sigma$, then
      $2d \leq 2\sigma \leq \delta$; by definition,
      $\delta = \delta^{\ast} + d_{n} -1$, so
      $d-d_{n} < \delta^{\ast} - d$; furthermore,
      $\delta^{\ast} - d < \delta^{\ast} - (\sigma^{\ast} +
      \mu^{\ast}) = \sigma^{\ast}$,
      so in the end:
      \begin{equation}
        \label{eq:60}
        d-d_{n} < \delta^{\ast} - d < \sigma^{\ast}.
      \end{equation}
      Since by construction, $\delta^{\ast}$ is divisible by
      $w_{n-1}$, the same reasoning as before yields that
      \begin{equation}
        \label{eq:61}
        a'_{d} = a^{\ast}_{d} - a^{\ast}_{d-d_{n}} = a^{\ast}_{\delta^{\ast}-d} - a^{\ast}_{d-d_{n}}
      \end{equation}
      and
      \begin{equation}
        a'_{d} \,
        \begin{cases}
          > 0 & \text{if both $\delta^{\ast} - d$ and $d-d_{n}$ are divisible by $w_{n-1}$;} \\
          = 0 & \text{if neither $\delta^{\ast} - d$ nor $d-d_{n}$ is
            divisible by $w_{n-1}$.}
        \end{cases}
      \end{equation}
    \end{itemize}

Still assuming that $d_{n}$, let now $d \in \{\sigma+1, \dots, \sigma + \mu\}$, we want to prove that $a_{d}-a_{d-1}=0$.
If $\mu=0$ there is nothing to prove, so assume that $\mu =1$ and $d= \sigma+1$.
But then $\sigma+1-d_{n} = \delta-\sigma-d_{n} = \delta^{\ast} - \sigma$, and so by symmetry, $a'_{d} = a^{\ast}_{\sigma+1} - a^{\ast}_{\sigma+1-d_{n}} = 0$.\qedhere
  \end{itemize}
\end{proof}

\begin{rmk}
  The hypothesis that the weights are reverse chain-divisible is
  necessary.  As a counter-example, let $W=(3,2,2)$ and $D=(6,6,6)$.
  Then the Hilbert series of a complete intersection of $W$-degree $D$
  is illustrated in Figure~\ref{fig:cex-HS-reg-1}.  It is
  self-reciprocal, but the coefficients do not vary as predicted by
  Theorem~\ref{thm:HS-regseq}.
  
  The hypothesis that each of the $W$-degrees should be divisible by
  $w_{1}$ is also necessary.  As a counter-example, let $W=(4,2,1)$
  and $D=(8,8,2)$.  Then the Hilbert series of a complete intersection
  of $W$-degree $D$ is illustrated in Figure~\ref{fig:cex-HS-reg-2}: the
  width of the steps is greater than $w_{n-1}$.  Furthermore, following the
  proof, the parameters for this series should be defined by
  $\sigma=\floor{\delta/2}$ and $\mu = \delta \mod 2$, where
  $\delta = 11$, so that $\sigma=5$ and $\mu=1$.  However, we cannot
  reorder the degrees such that $d_{3} \geq d_{2} \geq d_{1}$, and we
  cannot deduce from statement~\eqref{eq:lemMS3} in
  Lemma~\ref{lem:lemMS} that $\sigma^{\ast}+\mu^{\ast} \leq \sigma$:
  indeed, we have $\sigma=4$ but $\sigma^{\ast}+\mu^{\ast}=6$.

  However, the fact that the Hilbert series is self-reciprocal for
  complete intersections is true even for general system of weights,
  and is a consequence of the Gorenstein property of complete
  intersections (see~\cite[Chap.~21]{eisenbud95}; this property is
  also central to the proof of Theorem~\ref{prop:froberg-n+1}).
\end{rmk}

\begin{figure}
  \centering
  \begin{tikzpicture}[x=0.7cm,y=0.7cm]

  \tikzset{seriesLine/.style={draw=none}}
  
  \draw [color=black!10,step=1] (-0.1,-0.1) grid (12.8,3.8);

  \drawSeries{{
      0/1, 1/0, 2/2, 3/1,
      4/3, 5/2, 6/2, 7/3,
      8/1, 9/2, 10/0, 11/1, 12/0
    }}{p}{0}{0}{0}

  \draw [axis] (0,0) -- (13,0) node [anchor=north] {$d$};
  \draw [axis] (0,0) -- (0,4)  node [anchor=east]  {$a_{d}$}; 

\end{tikzpicture}


  \caption{Hilbert series of a weighted homogeneous complete
    intersection with \hbox{$W=(3,2,2)$} and \hbox{$D=(6,6,6)$}}
  \label{fig:cex-HS-reg-1}
\end{figure}
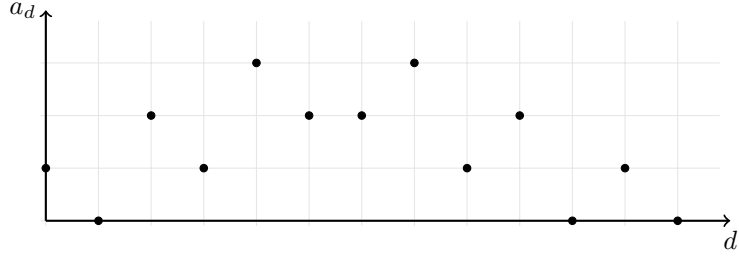

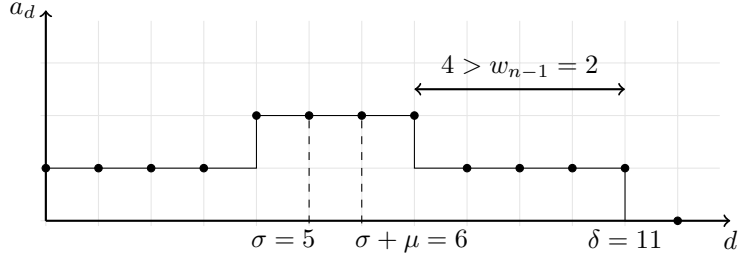
\begin{figure}
  \centering
  \begin{tikzpicture}[x=0.7cm,y=0.7cm]

  \draw [color=black!10,step=1] (-0.1,-0.1) grid (12.8,3.8);

  \drawSeries{{
      0/1, 1/1, 2/1, 3/1,
      4/2, 5/2, 6/2, 7/2,
      8/1, 9/1, 10/1, 11/1, 12/0
    }}{p}{0}{0}{0};

  \draw [axis] (0,0) -- (13,0) node [anchor=north] {$d$};
  \draw [axis] (0,0) -- (0,4)  node [anchor=east]  {$a_{d}$};

  \draw [<->,thick] (7,2.5) -- (11,2.5) node [pos=0.5,anchor=south] {$4 > w_{n-1}=2$};
  
  \node (d) [anchor=north] at (11,0) {$\delta = 11$};
  \node (s)                at (5,0)  {};
  \node (spm)              at (6,0)  {};

  \node [anchor=north east] at ($(s)+(0.3,0)$)   {$\sigma = 5$};
  \node [anchor=north west] at ($(spm)-(0.3,0)$) {$\sigma+\mu = 6$};
  
  \draw [dashed] ($(s)-(0,1pt)$)   -- (p5);
  \draw [dashed] ($(spm)-(0,1pt)$) -- (p6);

\end{tikzpicture}


  \caption{Hilbert series of a weighted homogeneous complete
    intersection with \hbox{$W=(4,2,1)$} and \hbox{$D=(8,8,2)$}}
  \label{fig:cex-HS-reg-2}
\end{figure}

\subsection{Degree of regularity of a weighted homogeneous complete
  intersection}
\label{sec:Degr-regul-quasi}

The degree of regularity of a zero-dimensional homogeneous regular
system is bounded by Macaulay's bound
\begin{equation}
  \label{eq:89}
  \dreg \leq \sum_{i=1}^{n} (d_{i} - 1) + 1,
\end{equation}
and, in practice, that bound is reached for generic systems.
The proof of this result uses the degree of the Hilbert series of the
system.  However, in the weighted case, the best result we can obtain
from the degree of the Hilbert series is \citep{FSV13}:
\begin{equation}
  \label{eq:91}
  \dreg[W] \leq \sum_{i=1}^{n} (d_{i} - w_{i}) + \max\{w_{j}\},
\end{equation}
and this bound is not sharp in general.  In particular, it appears
that this degree of regularity depends on the order we set on the
variables.

The following theorem is an improvement over the previous bound, under
the additional assumption that the system is in simultaneous Noether
position.  Recall that this property is generic, and that for reverse
chain-divisible systems of weights, it is always true for regular
sequences, up to a weighted homogeneous change of coordinates.

\begin{theorem}
  \label{thm:dreg}
  Let $W=(w_{1},\dots,w_{n})$ be a (not necessarily reverse
  chain-divisible) system of weights and $D=(d_{1},\dots,d_{n})$ be a
  strongly $W$-compatible system of $W$-degrees.  Further assume that
  for any $j \in \{2, \dots, n\}$, $d_{j} \geq w_{j-1}$.  Let
  $F=(f_{1},\dots,f_{n})$ be a system of $W$-homogeneous polynomials,
  with $W$-degree $D$, and assume that $F$ is in simultaneous Noether
  position for the variable ordering $X_{1} > X_{2} > \dots > X_{n}$.
  Then the $W$-degree of regularity of $F$ is bounded by
  \begin{equation}
    \label{eq:1}
    \dreg[W](F) \leq \sum_{i=1}^{n}(d_{i}-w_{i}) + w_{n}.
  \end{equation}
\end{theorem}
\begin{proof}
  We prove this by induction on $n$.  If $n=1$, we simply have one
  $W$-homogeneous polynomial to consider, and so $\dreg[W] = d_{1}$.

  So assume that $n>1$.  We consider the system $F^{\ast}$ defined by:
  \begin{equation}
    \label{eq:37}
    F^{\ast} = \big(f_{1}(X_{1},\dots,X_{n-1},0),\dots,f_{n-1}(X_{n-1},\dots,X_{n-1},0)\big).
  \end{equation}
  This system is $W^{\ast}$-homogeneous, for $W^{\ast} \coloneq (w_{1},\dots,w_{n-1})$.
  From the characterization~\ref{item:NP3} of Noether position, the
  sequence $F^{\ast}$ is in simultaneous Noether position.  As a
  consequence, the induction hypothesis applies to $F^{\ast}$, and the
  $W^{\ast}$-degree of regularity of $F^{\ast}$ is bounded by
  \begin{equation}
    \label{eq:dregstar}
    \dreg[W^{\ast}](F^{\ast}) \leq \sum_{i=1}^{n-1}(d_{i} - w_{i}) + w_{n-1}.
  \end{equation}

  Denote by $\delta$ the degree of the Hilbert series of $F$, that is
  $\delta = \sum_{i=1}^{n}(d_{i}-w_{i})$.  We want to prove that
  $\dreg \leq \delta +w_{n}$, i.e. that the Gröbner basis of $F$ need
  not contain any polynomial with $W$-degree greater than
  $\delta+w_{n}$.  Equivalently, let $\mu$ be a monomial with
  $W$-degree $d>\delta+w_{n}$, we will prove that $\mu$ is strictly
  divisible by a monomial in the initial ideal generated by $F$.
  
  Write $\mu = X_{n}^{\alpha}\cdot\mu'$, with
  $\mu' \in \K[X_{1},\dots,X_{n-1}]$, and proceed by induction on
  $\alpha$:
  \begin{itemize}
    \item if $\alpha = 0$, then $\mu \in \K[X_{1},\dots,X_{n-1}]$.  By
    assumption, $d_{n} \geq w_{n-1}$, hence:
    \begin{equation}
      \label{eq:dreg-geq-dregstar}
      \delta+w_{n}
      = \delta^{\ast} + w_{n-1} - w_{n-1} + d_{n} - w_{n} + w_{n}
      \geq \dreg^{\ast} + d_{n} - w_{n-1} \geq \dreg^{\ast};
    \end{equation}
    so $\mu$ has $W$-degree greater than $\dreg^{\ast}$, and by
    induction hypothesis, $\mu$ is strictly divisible by a monomial in
    the initial ideal generated by $F^{\ast}$;
    \item if $\alpha > 0$, then consider
    $\mu'' = X_{n}^{\alpha-1}\mu'$, it is a strict divisor of $\mu$.
    Furthermore, since $\deg(\mu) > \delta+w_{n}$,
    $\deg(\mu'') = \deg(\mu)-w_{n} > \delta$.  Recall that $\delta$ is
    by definition the degree of the Hilbert series of the ideal
    generated by $F$, so $\mu$ lies in that ideal.  \qedhere
  \end{itemize}
\end{proof}

\begin{remark}
  The hypothesis stating that for any $i$, $d_{i} \geq w_{i-1}$ is
  necessary.  For example, let $W=(2,1)$, $D=(2,1)$ and the system
  $F=(X,Y)$ in $\K[X,Y]$, it is $W$-homogeneous with $W$-degree $D$
  and in simultaneous Noether position.  This system has Hilbert
  series $1$ (the quotient vector span is generated by $\{1\}$), which
  has degree $\delta=0$.  But the Gröbner basis of the system is given
  by $F$ itself, and contains $X$, with $W$-degree $2$.

  More generally, without that hypothesis, we obtain the following
  bound for $\dreg[W](F)$:
  \begin{equation}
    \label{eq:35}
    \dreg[W](F) \leq \max\left\{ \sum_{i=1}^{k}(d_{i}-w_{i}) + w_{k}
     : k \in \{1,\dots,n\} \right\},
  \end{equation}
  and the proof is the same as that of Theorem~\ref{thm:dreg}, with
  the weaker induction hypothesis that
  $\dreg[W](F) \leq \max\left(\delta+w_{n},\dreg[W^{\ast}](F^{\ast})\right)$, which does not need inequality~\eqref{eq:dreg-geq-dregstar}.
\end{remark}

\begin{remark}
  We give examples of the behavior of both bounds in Table~\ref{tab:dreg}: we give the degree of regularity of a generic $W$-homogeneous system of $W$-degree $D$, and show how this degree of regularity varies if we change the order of the weights $W$.
\end{remark}

\begin{remark}
  Theorem~\ref{thm:dreg} gives an indication as to how to choose the
  order of the variables. Generically, in order to compute a \Wgrevlex
  Gröbner basis of the system, the complexity estimates will be better
  if we set the variables in decreasing weight order.
\end{remark}

While the new bound~\eqref{eq:1} is not sharp in full generality, it
is sharp whenever $w_{n}=1$. We conjecture that the sharp formula is
the following.

\begin{conjecture}
  Let $W=(w_{1},\dots,w_{n})$ be a system of weights, and
  $D=(d_{1},\dots,d_{n})$ a strongly $W$-compatible system of
  $W$-degrees. Let $F \in \K[X_{1},\dots,X_{n}]$ be a generic system
  of $W$-homogeneous polynomials. Let
  $\delta = \sum_{i=1}^{n} (d_{i}-w_{i})$ be the degree of the Hilbert
  series of $\langle F \rangle$, and let $d_{0}$ be defined as
  \begin{equation}
    \label{eq:128}
    d_{0} =
    \begin{cases}
      \delta + 1 & \textif \text{ there exists $i$ such that $w_{i}=1$} \\
      \delta - g & \otherwise,
    \end{cases}
  \end{equation}
  where $g$ is the Frobenius number of $W$ (that is, the greatest $W$-degree at which
  the set of monomials is empty). In other words, $d_{0}$ is the
  degree of the first ``unexpected'' zero coefficient in the Hilbert
  series (by definition of the degree in the first case, and by
  self-reciprocality of the Hilbert series in the second case).

  Then the degree of regularity of $F$ is the first multiple of
  $w_{n}$ greater than $d_{0}$:
  \begin{equation}
    \label{eq:130}
    \dreg = w_{n} \left\lceil \frac{d_{0}}{w_{n}} \right\rceil.
  \end{equation}
\end{conjecture}

\begin{table}
  \centering
  \begin{tabular}{ccccc}
    \toprule
    $W$ & $D$ & $\dreg$ & Bound~\eqref{eq:91} & Bound~\eqref{eq:1} \\
    \midrule
    $(3,2,1)$ & $(6,6,6)$ & 13 & 15 & 13 \\
    $(3,1,2)$ & $(6,6,6)$ & 14 & 15 & 14 \\
    $(1,2,3)$ & $(6,6,6)$ & 15 & 15 & 15 \\
  \end{tabular}
  \caption{Macaulay's bound on the degree of regularity of generic weighted homogeneous systems}
  \label{tab:dreg}
\end{table}

\section{Semi-regular systems}
\label{sec:Semi-regular-systems}

The study of systems with $m$ equations and $n$ unknowns, when
$m \leq n$, is reduced to the (generic) case of regular sequences,
sequences in Noether position or sequences in simultaneous Noether
position.

However, it is frequent in some applications that polynomial systems arise
with more equations than unknowns.  Experimentally, this usually makes
the resolution faster.  In the homogeneous case, this has been studied
through the notion of \emph{semi-regularity}.  This property extends
the regularity to the overdetermined case.  Fröberg's conjecture
\citep{Fro85} states that this property is generic, and as of today,
it is only known in a handful of cases (see~\cite[Thm.~1.5]{MS91} for a survey).

In this section, we give a definition of semi-regularity in the
weighted case, and some consequences on the Hilbert series and the
degree of regularity of the system.  Additionally, we show that
Fröberg's conjecture is true if $m = n+1$, as in the homogeneous case.

\subsection{Definitions and notations}
\label{sec:Definitions-1}

Let $n$ and $m$ be two integers, $m \geq n$, $W=(w_{1},\dots,w_{n})$ a
system of weights, and \hbox{$D=(d_{1},\dots,d_{n})$} a system of $W$-degrees.
Let $F=(f_{1},\dots,f_{m})$ be a system of $W$-homogeneous polynomials with $W$-degree $D$.
For any $i \in \{i, \dots, n\}$, write \hbox{$F_{i} = (f_{1},\dots,f_{i})$}.

\begin{definition}[Semi-regularity]
  We say that $F$ is \emph{semi-regular} if, for any
  $i \in \{1, \dots, m\}$ and for any $d \in \N$, the linear map given
  by the multiplication by $f_{i}$:
  \begin{equation}
    \label{eq:2}
    s_{i,d} : {\big(\K[X_{1},\dots,X_{n}]/\langle F_{i-1} \rangle\big)}_{d}
    \to^{\cdot f_{i}}
    {\big(\K[X_{1},\dots,X_{n}]/\langle F_{i-1} \rangle\big)}_{d+d_{i}}
  \end{equation}
  is full-rank (either injective or surjective).

  Furthermore, let
  \begin{equation}
    \label{eq:5}
    S_{D,W}(T) = \frac{\prod_{i=1}^{m}(1-T^{d_{i}})}{\prod_{i=1}^{n}(1-T^{w_{i}})} = \sum_{d=0}^{\infty}a_{d}T^{d}.
  \end{equation}
  We say that $F$ \emph{has a semi-regular Hilbert series} if the
  Hilbert series of $F$ is equal to $\lfloor S_{D,W}(T) \rfloor$, that
  is the series truncated at the first coefficient less than or equal
  to zero.
\end{definition}

The motivation behind these definitions is given by the following
classical result in the homogeneous case (see for
example~\cite[prop.~1]{Pardue2010}):
\begin{prop}
  \label{prop:HS-semireg-homo}
  If $W=(1,\dots,1)$, the following conditions are equivalent:
  \begin{enumerate}
    \item the system $F$ is semi-regular;
    \item for any $1 \leq i \leq n$, the system $F_{i}$ has a
    semi-regular Hilbert series.
  \end{enumerate}
\end{prop}

For weighted homogeneous systems, the converse implication ($2
\implies 1$) is still
true:
\begin{prop}
  Let $F$ be a $W$-homogeneous system such that, for any $1 \leq i \leq n$, the system $F_{i}$ has a semi-regular Hilbert series.
  Then $F$ is semi-regular.
\end{prop}
\begin{proof}
  We prove this by induction on the number $m$ of polynomials.  The
  initial case is $m=n$, and it is a direct consequence of the
  characterization of a regular sequence.

Assume $m>n$.
Write $R^{\ast} = \K[X_{1},\dots,X_{n}]/\langle f_{1},\dots,f_{m-1} \rangle$, and for any $d \in \N$, consider the multiplication map
  \begin{equation}
    \label{eq:31}
    s_{m,d} = R^{\ast}_{d} \to^{\cdot f_{m}} R^{\ast}_{d+d_{m}}
  \end{equation}
  Let $K_{m,d} = \ker(s_{m,d})$.  Write $S(T)$ the Hilbert series of
  $F$, $a_{d}$ its coefficient at degree $d$, $\delta$ its degree,
  $H(T) = \left( \prod_{i=1}^{m}(1-T^{d_{i}})\right) /\left(
   \prod_{i=1}^{n}(1-T^{w_{i}}) \right)$,
  $b_{d}$ its coefficient at degree $d$, and $S^{\ast}(T)$,
  $a^{\ast}_{d}$, $\delta^{\ast}$, $H^{\ast}(T)$ and $b^{\ast}_{d}$
  their counterparts with $m-1$ polynomials.  We know, from the exact
  sequence
  \begin{equation}
    \label{eq:32}
    0 \to K_{m,d} \to R_{d}^{\ast} \to^{s_{m,d}} R^{\ast}_{d+d_{m}} \to R_{d+d_{m}} \to 0
  \end{equation}
  that the following identity holds
  \begin{equation}
    \label{eq:33}
    a_{d+d_{m}} = a^{\ast}_{d+d_{m}} - a^{\ast}_{d} + \dim(K_{m,d}).
  \end{equation}
  We want to prove that either $a_{d+d_{m}}=0$ or $\dim(K_{m,d})=0$.
  Assume that $a_{d+d_{m}} > 0$, that means that
  $d + d_{m}\leq \delta$ and $a_{d+d_{m}}=b_{d+d_{m}}$, so:
  \begin{equation}
    \label{eq:34}
    \begin{aligned}[t]
      a_{d+d_{m}} &= a_{d+d_{m}} - a^{\ast}_{d} + \dim(K_{m,d}) \\
      &= b_{d+d_{m}} \\
      &= b^{\ast}_{d+d_{m}} - b^{\ast}_{d} \text{ by definition of $H(T)$} \\
      &= a^{\ast}_{d+d_{m}} - a^{\ast}_{d} \text{ since
        $\delta^{\ast} \geq \delta$.}
    \end{aligned}
  \end{equation}
  Thus we have $K_{m,d}=0$.
\end{proof}


\subsection{Hilbert series of a semi-regular sequence}
\label{sec:Cons-Hilb-seri}

In this section, we prove that for reverse chain-divisible systems of
weights, semi-regular sequences have a semi-regular Hilbert series.
First, we characterize the shape of semi-regular Hilbert series, by
extending Theorem~\ref{thm:HS-regseq} to the overdetermined case.

\begin{theorem}
  \label{thm:shape-HS-sr}
  Let $m \geq n \geq 0$ be two integers.  Let
  $W = (w_{1},\dots,w_{n})$ be a reverse chain-divisible system of
  weights, and let $D=(d_{1},\dots,d_{m})$ be a system of $W$-degrees
  such that $d_{1},\dots,d_{n}$ are all divisible by $w_{1}$.  Write
  \begin{equation}
    \label{eq:14}
    S_{D,W}(T) = \frac{\prod_{i=1}^{m}(1-T^{d_{i}})}{\prod_{i=1}^{n}(1-T^{w_{i}})} = \sum_{d=0}^{\infty}a_{d}T^{d}.
  \end{equation}
  Then there exist $W$-degrees $\sigma$, $\delta$ such that
  \begin{gather}
    \label{eq:93}\tag{$\sigma1$}
    \forall\, d \in \{1,\dots,\sigma\},\; a_{d} \geq a_{d-1} \\
    \label{eq:94}\tag{$\sigma2$}
    a_{\sigma} > a_{\sigma-1} \\
    \label{eq:95}\tag{$\sigma3$}
    \forall\, d \in \{\sigma+1,\dots,\delta\},\; a_{d} \leq a_{d-1} \\
    \label{eq:96}\tag{$\delta1$}
    a_{\delta} > 0, \, a_{\delta+1} \leq 0.
  \end{gather}
  Furthermore, if $m > n$, let $D^{\ast} = (d_{1},\dots,d_{m-1})$ and
  define $\delta^{\ast}$ as above for the series $S_{D^{\ast},W}$.
  Then the following statements hold:
  \begin{equation}
    \label{eq:92}\tag{$\delta2$}
    \begin{cases}
      \forall\, d \in \{\delta+1,\dots,\delta^{\ast}\}, \,a_{d} \leq 0 & \text{if } n=0 \\
      \forall\, d \in \{\delta+1,\dots,\delta^{\ast} + d_{m}\}, \,a_{d} \leq 0 & \text{if } n>0. \\
    \end{cases}
  \end{equation}
  If $n>0$, let $W^{\ast}=(w_{1},\dots,w_{n-1})$, and let $\delta'$ be
  the degree of $\lfloor S_{D,W^{\ast}}(T) \rfloor$.  If $n=0$, let
  $\delta' = 0$.  Then the following equality holds
  \begin{equation}
    \label{eq:97}\tag{$\sigma4$}
    \sigma = \delta'. 
  \end{equation}
\end{theorem}

\begin{proof}
  We prove the theorem by induction on $n$, and for any given $n$,
  by induction over $m$.  The base cases are:
  \begin{itemize}
    \item $n=0$, $m \geq 0$: then
    $S_{D,W}(T) = 1 - a_{k}T^{k} + O(T^{k+1})$ with $a_{k} > 0$, and
    we can conclude, taking $\delta=0$ and $\sigma=0$.
    \item $n=m > 0$: then this is a consequence of
    Theorem~\ref{thm:HS-regseq} (shape of the Hilbert series of a
    complete intersection).
  \end{itemize}
    
  Assume that $m > n > 0$.  Let $D^{\ast}=(d_{1},\dots,d_{m-1})$,
  $W^{\ast}=(w_{1},\dots,w_{n-1})$, and write:
  \begin{gather}
    \label{eq:36}
    S(T) \coloneq S_{D,W}(T) = \sum_{d=0}^{\infty}a_{d}T^{d}; \\
    \label{eq:39}
    S^{\ast}(T) \coloneq S_{D^{\ast},W}(T) =
    \sum_{d=0}^{\infty}a^{\ast}_{d}T^{d}.
  \end{gather}
  The derivatives of these series are
  \begin{gather}
    \label{eq:45}
    \Delta S(T) = S_{D,W^{\ast}}(T) = \sum_{d=0}^{\infty}a'_{d}T^{d}; \\
    \label{eq:47}
    \Delta S^{\ast}(T) = S_{D^{\ast},W^{\ast}}(T) =
    \sum_{d=0}^{\infty}a'^{\ast}_{d}T^{d}.
  \end{gather}
  Furthermore, let $w=w_{n-1}$,
  $\bar{W^{\ast}} = (w_{1}/w, \dots, w_{n-1}/w)$ and
  $\bar{D^{\ast}} = (d_{1}/w, \dots, d_{n-1}/w)$, and consider the
  series
  \begin{gather}
    \label{eq:100}
    \bar{\Delta S}(T) = S_{\bar{D},\bar{W^{\ast}}}(T) = \sum_{d=0}^{\infty}\bar{a'}{d}T^{d} ; \\
    \label{eq:98}
    \bar{\Delta S^{\ast}}(T) = S_{\bar{D^{\ast}},\bar{W^{\ast}}}(T)
    \sum_{d=0}^{\infty}\bar{a'^{\ast}}{d}T^{d}.
  \end{gather}
  In particular,
  \begin{equation}
    \label{eq:99}
    \Delta S(T) = \bar{\Delta S}(T^{w}) \text{ and } \Delta S^{\ast}(T) = \bar{\Delta S^{\ast}}(T^{w}).
  \end{equation}

  All the series $S^{\ast}$, $\bar{\Delta S}$ and
  $\bar{\Delta S^{\ast}}$ satisfy the induction hypothesis.  The
  $W$-degrees for which the coefficients of $S^{\ast}$ satisfy
  properties~\eqref{eq:93}-\eqref{eq:97}
  and~\eqref{eq:96}-\eqref{eq:92} are denoted by $\sigma^{\ast}$ and
  $\delta^{\ast}$.  We write $\bar{\sigma'}$, $\bar{\delta'}$,
  $\bar{\sigma'^{\ast}}$, $\bar{\delta'^{\ast}}$ the respective values
  of the $W$-degrees for which these properties apply to
  $\bar{\Delta S}$ and $\bar{\Delta S^{\ast}}$.

  From $S(T) = (1-T^{d_{m}})S^{\ast}(T)$, we deduce the recurrence
  relation
  \begin{equation}
    \label{eq:68}
    a_{d} = a^{\ast}_{d} - a^{\ast}_{d-d_{m}}.
  \end{equation}
  Since $S^{\ast}$ satisfies the induction hypothesis, we know that
  there exists a degree $\delta$ such that
  \begin{equation}
    \label{eq:sttmt-delta}
    \begin{cases}
      \forall\, d \in \{0, \dots, \delta\} & a^{\ast}_{d} > a^{\ast}_{d-d_{m}} \\
      \forall\, d \in \{\delta+1, \dots, \delta^{\ast}+d_{m}\} &
      a^{\ast}_{d} \leq a^{\ast}_{d-d_{m}}.
    \end{cases}
  \end{equation}
  This proves statements~\eqref{eq:96} and~\eqref{eq:92}.  As a side
  result, since $a^{\ast}_{\delta} > a^{\ast}_{\delta-d_{m}}$, we also
  deduce that
  \begin{equation}
    \label{eq:delta-minus-dm-leq-sigmastar}
    \delta - d_{m} < \sigma^{\ast}.
  \end{equation}

  Let $\sigma = \delta' $, we prove that it satisfies
  equations~\eqref{eq:93},~\eqref{eq:94} and~\eqref{eq:95}.  We need
  to evaluate the sign of $a_{d} - a_{d-1}$, depending on $d$.  The
  generating series of $a_{d}-a_{d-1}$ is:
  \begin{equation}
    \label{eq:126}
    (1-T)S(T)
    = (1-T) \cdot \frac{\prod_{i=1}^{m}(1-T^{d_{i}})}{\prod_{i=1}^{n}(1-T^{w_{i}})}
    = \Delta S(T) \text{ since $w_{n} =1$.}
  \end{equation}
  In other words, $a_{d} \geq a_{d-1}$ if and only if $a'_{d} \geq 0$,
  which proves statements~\eqref{eq:93} and~\eqref{eq:94}, by
  definition of $\delta'$:
  \begin{gather}
    \label{eq:sttmt-sigma1}
    \forall\, d \in \{0,\dots,\sigma\},\;  a_{d} - a_{d-1} = a'_{d} \geq 0 \\
    \label{eq:13}
    a_{\sigma} - a_{\sigma-1} = a'_{\sigma} = a'_{\delta'} > 0.
  \end{gather}
  To finish the proof, we need to prove that for any
  $d \in \{\delta'+1,\dots,\delta\}$, $a'_{d}\leq 0$.

  From the induction hypothesis (statement~\eqref{eq:97}) applied to
  $S^{\ast}$, we know that $\delta'^{\ast}=\sigma^{\ast}$.  Moreover,
  statement~\eqref{eq:92} from the induction hypothesis applied to
  $\bar{\Delta S}$ yields that:
  \begin{equation}
    \label{eq:11}
    \forall\, \bar{d} \in \{\bar{\delta'}+1, \dots, \bar{\delta'^{\ast}}\},\;
    \bar{a'_{\bar{d}}} \leq 0.
  \end{equation}
  As a consequence, since
  $\sigma^{\ast} = \delta'^{\ast} = w\bar{\delta'^{\ast}}$:
  \begin{equation}
    \label{eq:sttmt-sigma2}
    \forall\, d \in \{\delta'+1,\dots,\sigma^{\ast}\},\;
    a'_{d} =
    \begin{cases}
      \bar{a'_{\bar{d}}} \leq 0 & \text{if $d = w\bar{d}$;} \\
      0 & \text{otherwise.}
    \end{cases}
    \tag{\romannumeral 1}
  \end{equation}
  Now assume that $\sigma^{\ast} < d \leq \delta$.  We can write
  $a'_{d}$ as
  \begin{align}
    \label{eq:73}
    a'_{d} = a_{d} - a_{d-1} &= a^{\ast}_{d} - a^{\ast}_{d-d_{m}} - a^{\ast}_{d-1} + a^{\ast}_{d-d_{m}-1} \\
    \label{eq:74}
    &= (a^{\ast}_{d} - a^{\ast}_{d-1}) - (a^{\ast}_{d-d_{m}} -
    a^{\ast}_{d-d_{m}-1}) = a'^{\ast}_{d} - a'^{\ast}_{d-d_{m}}.
  \end{align}

  Since $a_{d} \leq a^{\ast}_{d}$ for any $d$, we necessarily have
  $\delta \leq \delta^{\ast}$, hence
  $\sigma^{\ast} < d \leq \delta^{\ast}$.  So by induction hypothesis
  (statement~\eqref{eq:95}), we know that
  $a^{\ast}_{d} - a^{\ast}_{d-1} \leq 0$.  Additionally, 
  equation~\eqref{eq:delta-minus-dm-leq-sigmastar} and induction
  hypothesis (statement~\eqref{eq:93}) together yield that
  $a^{\ast}_{d-d_{m}} - a^{\ast}_{d-d_{m}-1} \geq 0$, so we conclude
  that
  \begin{equation}
    \label{eq:sttmt-sigma3}\tag{\romannumeral 2}
    \forall\, d \in \{\sigma^{\ast}+1, \dots, \delta\}, \;a'_{d} \leq 0.
  \end{equation}
  And so, sticking the ranges of statements~\eqref{eq:sttmt-sigma2}
  and~\eqref{eq:sttmt-sigma3} together, we prove
  statement~\eqref{eq:95} which completes the proof.
\end{proof}

Using this description of semi-regular Hilbert series, we now prove
that for reverse chain-divisible systems of weights, semi-regular
sequences have a semi-regular Hilbert series.  As an illustration,
Figure~\ref{fig:shape-HS-semireg} shows the coefficient of a
semi-regular Hilbert series.  The black dots correspond to the actual
coefficients, and the gray dots are the coefficients which were
truncated away.

\begin{figure}
  \centering
  \begin{tikzpicture}[x=0.7cm,y=0.7cm]
  \draw [color=black!10,step=1] (-0.1,-1.5) grid (10.8,2.8);

  \begin{scope}[color=black!30]
    \drawSeries{{9 / -1 , 10 / -1}}{pr}{8}{0}{10};
    \draw (pr11) -- (10.8,-1);
  \end{scope}

  \begin{scope}[color=black]
    \drawSeries{{%
        0 / 1 , 1 / 1 , 2 / 1 , 3 / 2 , 4 / 2 ,
        5 / 2 , 6 / 1 , 7 / 1 , 8 / 1 ,
        9 / 0 , 10 / 0 }}{p}{0}{0}{0};
    \draw    (p10)  -- (10.8,0);
  \end{scope}

  \draw [axis] (0,0) -- (11,0) node [anchor=north] {$d$};
  \draw [axis] (0,-1.5) -- (0,3) node [anchor=east] {$a_{d}$};;

  \node [anchor=north] (s)   at (3,0) {$\sigma = 3$};
  \node [anchor=north] (d)   at (8,0) {$\delta = 8$};

  \draw [dashed] (s) -- (p3);
  \draw [dashed] (d) -- (p8);

\end{tikzpicture}


  \caption{Shape of the Hilbert series of a semi-regular
    $W$-homogeneous sequence with $W=(3,3,1)$ and $D=(12,9,6,6,3)$}
  \label{fig:shape-HS-semireg}
\end{figure}
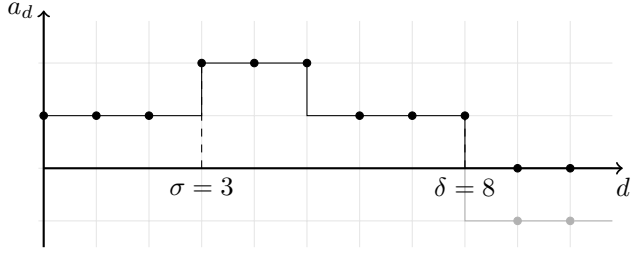

\begin{theorem}
  \label{thm:HS-semireg}
  Let $m \geq n \geq 0$ be two integers, $W=(w_{1},\dots,w_{n})$ be a
  reverse chain-divisible system of weights and
  $D=(d_{1},\dots,d_{m})$ be a system of $W$-degrees such that
  $d_{1},\dots,d_{n}$ are all divisible by $w_{1}$.  Let
  $F=(f_{1},\dots,f_{m})$ be a system of $W$-homogeneous polynomials,
  with respective $W$-degree $D$.  If $F$ is a semi-regular sequence,
  then $F$ has a semi-regular Hilbert series.
\end{theorem}
\begin{proof}
  We proceed by induction on $m$.  If $m=n$, then the result is a
  consequence of the characterization of regular sequences.
  
Assume that $m > n$.
We consider the series $S(T)=S_{D,W}(T)$ with generic coefficient $a_{d}$, \hbox{$S^{\ast}(T)=S_{D^{\ast},W}(T)$} with generic coefficient $a^{\ast}_{d}$, $H(T)$ the Hilbert series of $F$ with generic coefficient $b_{d}$, and $H^{\ast}(T)$ the Hilbert series of $F^{\ast} \coloneq (f_{1},\dots,f_{m-1})$ with generic coefficient $b^{\ast}_{d}$.
By induction hypothesis, $H^{\ast}(T) = \lfloor S^{\ast}(T) \rfloor$.
And since $F$ is semi-regular, we have the exact sequence
  \begin{equation}
    \label{eq:69}
    0 \to K_{m,d} \to R_{d}^{\ast} \to^{s_{m,d}} R^{\ast}_{d+d_{m}} \to R_{d+d_{m}} \to 0
  \end{equation}
  where $R = \KK[X_{1},\dots,X_{n}]/\langle F \rangle$ and
  $R^{\ast} = \KK[X_{1},\dots,X_{n}]/\langle F^{\ast} \rangle$.  As a
  consequence, for any $d \geq 0$, the coefficient $b_{d}$ satisfies
  the recurrence relation:
  \begin{equation}
    \label{eq:70}
    b_{d+d_{m}} = b^{\ast}_{d+d_{m}} - b^{\ast}_{d} + \dim(K_{m,d})
  \end{equation}
  where either $K_{m,d} = 0$ or $b_{d+d_{m}}=0$.  Since $s_{m,d}$ is
  defined from a space of dimension $b^{\ast}_{d}$ to a space of
  dimension $b^{\ast}_{d+d_{m}}$, this can be rephrased as
  \begin{equation}
    \label{eq:72}
    b_{d} = \max\left( b^{\ast}_{d}-b^{\ast}_{d-d_{m}},0 \right).
  \end{equation}
  From Theorem~\ref{thm:shape-HS-sr} applied to $S(T)$, there exists
  $\delta$ such that
  \begin{equation}
    \label{eq:77}
    \forall d \in \{0, \dots, \delta\},\, a_{d}=a^{\ast}_{d}-a^{\ast}_{d-d_{m}} > 0.
  \end{equation}
  Furthermore, the induction hypothesis shows that there exists a
  degree $\delta^{\ast}$ such that
  \begin{gather}
    \label{eq:78}
    \forall d \in \{0, \dots, \delta^{\ast}\},\, a^{\ast}_{d}=b^{\ast}_{d} > 0 \\
    \label{eq:79}
    \forall d > \delta^{\ast}, b^{\ast}_{d} = 0,
  \end{gather}
  and that $\delta^{\ast}$ is defined as in
  Theorem~\ref{thm:shape-HS-sr}.  In particular, it implies that
  $\delta^{\ast} \geq \delta$.
  
  We shall prove that the Hilbert series $H$ of $F$ is equal to $S$,
  truncated at degree $\delta$.  Let $d \geq 0$:
  \begin{itemize}
    \item if $0 \leq d \leq \delta \leq \delta^{\ast}$:
    \begin{align}
      b_{d} & = b^{\ast}_{d} - b^{\ast}_{d_{m}} \text{ since $d \leq \delta$}\\
      &= a^{\ast}_{d} - a^{\ast}_{d_{m}} \text{ since $d \leq \delta^{\ast}$} \\
      &= a_{d}
    \end{align}
    \item if $\delta < d$:
    \begin{align}
      b_{d} &= \max\left( b^{\ast}_{d}-b^{\ast}_{d-d_{m}},0 \right) \\
      &=0 \text{ since $b^{\ast}_{d} = 0$ and
        $b^{\ast}_{d-d_{m}} \geq 0$}
    \end{align}
  \end{itemize}
  And since $a_{\delta+1} \leq 0$, this proves that
  \begin{equation}
    \label{eq:83}
    H(T) = \lfloor S(T) \rfloor.\qedhere
  \end{equation}
\end{proof}

Another consequence of Theorem~\ref{thm:shape-HS-sr} is an explicit
value for the degree $\delta$ of the Hilbert series of an ideal
defined by a semi-regular sequence with $m=n+1$ polynomials in $n$
variables.  In the homogeneous case, it is known %
that this degree is bounded by
\begin{equation}
  \label{eq:103}
  \delta = \min\left( \sum_{i=1}^{n} d_{i} -n,
   \left\lfloor
    \frac{\sum_{i=1}^{n+1} d_{i} -n}{2}
   \right\rfloor \right).
\end{equation}

\begin{prop}
  \label{prop:dreg-semireg-n+1}
  Let $n$ be a positive integer, and $m=n+1$.  Let
  $W=(w_{1},\dots,w_{n})$ be a system of weights, and
  $F=(f_{1},\dots,f_{m})$ a system of $W$-homogeneous polynomials, and
  assume that the hypotheses of Theorem~\ref{thm:HS-semireg} are
  satisfied. For all $i \in \{1,\dots,m\}$, let
  $d_{i} \coloneq \deg_{W}(f_{i})$. Then the degree $\delta$ of the
  Hilbert series of $\langle F \rangle$ is given by:
  \begin{equation}
    \label{eq:104}
    \delta = \min\left( \sum_{i=1}^{n} d_{i} - \sum_{i=1}^{n}w_{i},
     \left\lfloor
      \frac{\sum_{i=1}^{n+1} d_{i} - \sum_{i=1}^{n}w_{i}}{2}
     \right\rfloor \right).
  \end{equation}
\end{prop}
\begin{proof}
  Consider the system of weights $W^{+}=(w_{1},\dots,w_{n},1)$, and
  the series $S_{D,W^{+}}$ as defined in
  Theorem~\ref{thm:shape-HS-sr}.  It satisfies the hypotheses of
  Theorem~\ref{thm:HS-regseq}, which implies that its coefficients are
  increasing up to degree
  \begin{equation}
    \label{eq:106}
    \sigma^{+} = \min\left( \sum_{i=1}^{n} d_{i} - \sum_{i=1}^{n}w_{i},
     \left\lfloor
      \frac{\sum_{i=1}^{n+1} d_{i} - \sum_{i=1}^{n}w_{i}}{2}
     \right\rfloor \right).
  \end{equation}
  Theorem~\ref{thm:shape-HS-sr} (statement~\eqref{eq:97}) states that
  the degree $\delta$ of the Hilbert series of $\langle F \rangle$
  satisfies
  \begin{equation}
    \label{eq:107}
    \delta = \sigma^{+},
  \end{equation}
  hence the result.
\end{proof}

\subsection{Asymptotic analysis of the degree of regularity of
  weighted homogeneous semi-regular sequences}
\label{sec:Asympt-analys-degr}

In this section, we show how the results from~\cite{BFSY05} and
\cite[Chap.~4]{Bar04} about the degree of regularity of semi-regular
homogeneous sequences can be adapted to the weighted case.

\begin{theorem}
  Let $k$, and $n$ be non-negative integers and let $m \coloneq n+k$.  Let
  $w_{0}$ and $d_{0}$ be non-negative integers such that $w_{0} \divides d_{0}$.
  Consider the system of $n$ weights $W=(w_{0},\dots,w_{0},1)$ and the system
  of $m$ $W$-degrees $D=(d_{0},\dots,d_{0})$.  Let
  $F=(f_{1},\dots,f_{m}) \subset A=\KK[X_{1},\dots,X_{n}]$ be a
  semi-regular sequence of weighted homogeneous polynomials with
  $W$-degree $D$.  Then the asymptotic
  developement of the $W$-degree of regularity $\dreg$ of $F$ as $n$
  tends to infinity is given by
  \begin{equation}
    \label{eq:dreg-asymptotic-W}
    \dreg = n \left( \frac{d_{0}-w_{0}}{2} \right)
        - \alpha_{k}\sqrt{n \left(\frac{d_{0}^{2}-w_{0}^{2}}{6}\right) }
        + \Ocomp{n^{1/4}}.
  \end{equation}
\end{theorem}

\begin{remark}
  In the non-weighted case, this asymptotic developement is
  \begin{equation}
    \label{eq:dreg-asymptotic-1}
    \dreg = n \left(\frac{d_{0}-1}{2}\right)
            - \alpha_{k}\sqrt{n\left(\frac{d_{0}^{2} - 1}{6}\right)}
            + \Ocomp{n^{1/4}}.
  \end{equation}
  Overall, the bound is improved by $\Ocomp{nw_{0}} = \Ocomp{\sum w_{i}}$.
\end{remark}

\begin{proof}
  Let $I \coloneq \langle F \rangle$, the Hilbert series of $A/I$ is
  given by
  \begin{equation}
    \label{eq:105}
    \HS_{A/I}(T) = \left\lfloor \frac{{(1-T^{d_{0}})}^{m}}%
     {{(1-T^{w_{0}})}^{n-1}(1-T)}\right\rfloor. 
  \end{equation}
  Write
  \begin{gather}
    \label{eq:111}
    H(T) = \frac{{(1-T^{d_{0}})}^{m}}{{(1-T^{w_{0}})}^{n-1}(1-T)}
    = \sum_{d=0}^{\delta} a_{d}T^{d}; \\
    \label{eq:108}
    H^{\ast}(T) = \frac{{(1-T^{d_{0}/w_{0}})}^{m-1}}{{(1-T)}^{n-1}} =
    \sum_{d=0}^{\delta} a^{\ast}_{d}T^{d},
  \end{gather}
  these series are related through
  \begin{equation}
    \label{eq:109}
    H(T)
    = H^{\ast}(T^{w_{0}})\cdot \frac{1-T^{d_{0}}}{1-T}
    = H^{\ast}(T^{w_{0}})\cdot \left( 1 + T + \dots + T^{d_{0}-1} \right).
  \end{equation}
  For the coefficients, it means that, for any $d$ in $\N$:
  \begin{equation}
    \label{eq:112}
    a_{d} = a^{\ast}_{\floor{d/w_{0}}} + \dots + a^{\ast}_{\ceil{(d-d_{0}+1)/w_{0}}}
  \end{equation}

  The series $H^{\ast}$, if truncated before its first non-positive
  coefficient, is the Hilbert series of a semi-regular
  $\WOne$-homogeneous sequence of $m-1$ polynomials in $n-1$
  variables, with degree $d_{0}/w_{0}$.  Let $\delta^{\ast}$ be the degree
  of this truncated series, so that $\delta^{\ast}+1$ is an upper
  bound for the degree of regularity of such a sequence.

  Statement~\eqref{eq:92} of Theorem~\ref{thm:shape-HS-sr} states
  that:
  \begin{equation}
    \label{eq:114}
    \forall\, d \in \{\delta^{\ast}+1,\dots,\delta^{\ast}+d_{0}/w_{0}\}, \,a^{\ast}_{d}\leq 0.
  \end{equation}
  Let $\delta_{0} \coloneq w_{0}\delta^{\ast}+d_{0}$, we have
  \begin{equation}
    \label{eq:116}
    \delta^{\ast} < \frac{\delta_{0}-d_{0}+1}{w_{0}} \leq \delta^{\ast} +1
  \end{equation}
  and
  \begin{equation}
    \label{eq:117}
    \frac{\delta_{0}}{w_{0}} = \delta^{\ast} + \frac{d_{0}}{w_{0}}, 
  \end{equation}
  and as a consequence
  \begin{equation}
    \label{eq:115}
    a_{\delta^{0}} = a^{\ast}_{\floor{\delta_{0}/w_{0}}} + \dots + a^{\ast}_{\ceil{(\delta_{0}-d_{0}+1)/w_{0}}} 
    \leq 0.
  \end{equation}
  In other words, the degree of regularity $\dreg$ of $F$ is bounded
  by
  \begin{equation}
    \label{eq:119}
    w_{0}\delta^{\ast} < \dreg \leq \delta_{0} = w_{0}\delta^{\ast}+d_{0}.
  \end{equation}

  The degree $\delta^{\ast}$ is the degree of the Hilbert series of a
  homogeneous semi-regular sequence, and as such, it follows the
  asymptotic estimates proved in~\cite[Chap.~4]{Bar04}.  For example,
  in our setting where $k$ is an integer and $m=n+k$, the asymptotic
  developement of $\delta^{\ast}$ when $n$ tends to infinity is given
  by%
  \begin{equation}
    \label{eq:120}
    \delta^{\ast} +1
    = n \,\frac{d_{0}/w_{0}-1}{2} - \alpha_{k}\sqrt{n\, \frac{{\left(d_{0}/w_{0}\right)}^{2} - 1}{6}} + \Ocomp{n^{1/4}}
    \end{equation}
  where $\alpha_{k}$ is the largest root of the $k$'th Hermite's
  polynomial.  \footnote{In~\cite[Chap.~4]{Bar04}, the remainder
    $\Ocomp{n^{1/4}}$ was written as $\ocomp{\sqrt{n}}$.  However, it
    appears that in the proof, this $\ocomp{\sqrt{n}}$ is a rewriting
    of $\sqrt{n} \cdot \Ocomp{\sqrt{\Delta z}}$, where
    $\Delta z = \Ocomp{1/\sqrt{n}}$.\ifdraft{{\color{gray}[Footnote
        pour la version finale]}}{}}

  As a consequence,
  \begin{align}
    \label{eq:121}
    \dreg
    &= w_{0}\delta^{\ast} + \Ocomp{1} \\
    \label{eq:122}
    &= w_{0}\left( n\, \frac{d_{0}/w_{0}-1}{2} - \alpha_{k}\sqrt{n\,
       \frac{{\left(d_{0}/w_{0}\right)}^{2} - 1}{6}} +\Ocomp{n^{1/4}}
    \right) + \Ocomp{1} \\
    \label{eq:123}
    &= n \,\frac{d_{0}-w_{0}}{2} - \alpha_{k}\sqrt{n\, \frac{d_{0}^{2} -
        w_{0}^{2}}{6}} + \Ocomp{n^{1/4}}. \qedhere
  \end{align}
\end{proof}

\subsection{Fröberg's conjecture}
\label{sec:Frobergs-conjecture}

Fröberg's conjecture states that homogeneous semi-regular sequences are generic among sequences of fixed degree.
The fact that semi-regularity is a Zariski-open condition is a known fact (the proof is the same as for regularity), so the conjecture states that for any system of degrees, there exists a semi-regular homogeneous sequence with these degrees.

This conjecture extends naturally to the weighted case.  In this case,
semi-regularity is still a Zariski-open condition.

We extend here one known result from the homogeneous case (see for
example~\cite{ReidRobertsRoitman1991}), stating that Fröberg's
conjecture is true in characteristic $0$ if $m=n+1$.  We follow the
proof given in~\cite{ReidRobertsRoitman1991}.

\begin{prop}
  \label{prop:froberg-n+1}
  Let $m = n+1$, $W=(w_{1},\dots,w_{n})$ a reverse chain-divisible
  system of weights, $D=(d_{1},\dots,d_{n})$ a strongly $W$-compatible
  system of degrees, and $d_{n+1}$ an integer divisible by $w_{1}$.
  Write
  $f_{n+1} = {(X_{1} + X_{2}^{w_{1}/w_{2}} + \dots +
    X_{n}^{w_{1}})}^{d_{n+1}/w_{1}}$,
  then the sequence
  \hbox{$F=(X_{1}^{d_{1}/w_{1}},\dots,X_{n}^{d_{n}/w_{n}},f_{n+1})$} is
  semi-regular.
\end{prop}

\begin{lemma}
  Let $f$ be a polynomial such that
  \begin{equation}
    \label{eq:15}
    f\cdot f_{n+1} = 0 \textin A = \KK[X_{1},\dots,X_{n}]/(X_{1}^{d_{1}/w_{1}},\dots,X_{n}^{d_{n}/w_{n}}).
  \end{equation}
  Let $\delta = \sum_{i=1}^{n}(d_{i}-w_{i})$, then we have
  \begin{equation}
    \label{eq:23}
    \deg_{W}(f) \geq \frac{\delta - d_{n+1} +1}{2}.
  \end{equation}
\end{lemma}
\begin{proof}
  If the $W$-degree of $f$ is $0$, that means that ${(X_{1} + X_{2}^{w_{1}/w_{2}} + \dots + X_{n}^{w_{1}})}^{d_{n+1}/w_{1}}=0$ in $A$.
  Assume that $\deg_{W}(f) < (\delta-d_{n+1}+1)/2$, that means that \hbox{$\delta-d_{n+1}+1 \geq 1$}, so $\delta \geq d_{n+1}$.
  Consider the expansion of $f_{n+1}$, all coefficients are nonzero since the base field has characteristic $0$.
  Its support is the set of monomials of degree $d_{n+1}$.
  Since $d_{n+1} \leq \delta$, $\dim{\left( \KK[\mathbf{X}]/\langle f_{1},\dots,f_{n}\rangle \right)}_{d_{n+1}} > 0$, which means that there exists at least one monomial with $W$-degree $d_{n+1}$ which does not lie in the initial ideal of $\langle f_{1},\dots,f_{n}\rangle$.
  As a consequence, $f$ is non-zero in the quotient.

  Now assume that $\deg_{W}(f) > 0$.  Write
  $B =
  \KK[X_{2},\dots,X_{n}]/(X_{2}^{d_{2}/w_{2}},\dots,X_{n}^{d_{n}/w_{n}})$,
  $X=X_{1}$, $R=B[X]$, $d=d_{1}/w_{1}$, such that $A=R/X^{d}$.  Let
  $S = (X + X_{2}^{w_{1}/w_{2}} + \dots + X_{n}^{w_{1}})$, and let $F$
  be a weighted homogeneous polynomial in $R$ with image $f$ in $A$.
  By assumption, there exists $G \in R$ such that
  $S^{d_{n+1}/w_{1}}\cdot F=G\cdot X^{d}$.  Derive this equality along
  $X$ to obtain:
  \begin{equation}
    \label{eq:24}
    mS'S^{d_{n+1}/w_{1} -1}F + S^{(d_{n+1})/w_{1}}F = dG'X^{d-1} + G'X^{d}
  \end{equation}
  or, modulo $X^{d-1}$
  \begin{align}
    \label{eq:25}
    S^{d_{n+1}-1}(mF + SF') \equiv 0 \mod X^{d-1}
    & \implies S^{d_{n+1}/w_{1}}(mF + SF') \equiv 0 \mod X^{d-1} \\
    \label{eq:27}
    & \implies S^{d_{n+1}/w_{1}+1}F' \equiv mFS^{d_{n+1}/w_{1}} \equiv
    0 \mod X^{d-1}
  \end{align}
  Since $X=X_{1}$ has weight $w_{1}$, $F'$ is $W$-homogeneous with
  $W$-degree $\deg_{W}(f)-w_{1}$, and we can use the induction
  hypothesis on $F' \mod X \in A$ and $\deg(F) = d_{n+1}+w_{1}$ to
  deduce:
  \begin{align}
    \label{eq:28}
    \deg_{W}(f) &= \deg_{W}(F) = \deg_{W}(F') +1 \\
    \label{eq:29}
    & \geq \frac{(\delta-1)-(d_{n+1}+1)+1}{2} +1 \\
    \label{eq:30}
    & \geq \frac{\delta - d_{n+1} +1}{2}.
  \end{align}
\end{proof}

\begin{proof}[Proof of the proposition]
  The proof given in~\cite[before prop.~7]{ReidRobertsRoitman1991}
  still holds in the weighted case.
\end{proof}

\section{Taking into account a weighted homogeneous structure
  when computing Gröbner bases}
\label{sec:Cons-comp}

\subsection{Weighted homogeneous systems}
\label{sec:Quasi-homog-syst}

Let $n,m$ be two integers, let $W=(w_{1},\dots,w_{n})$ be a system of
weights, and let
$F=(f_{1},\dots,f_{m})$ in $\KK[X_{1},\dots,X_{n}]$ be a system of
weighted homogeneous polynomials.

In order to solve the system $F$, we need to compute a Gröbner basis
for some monomial order, usually an elimination order or the
lexicographical order.  The usual strategy for that purpose is to
perform the computation in two steps, first computing a Gröbner basis
for some ``easy'' order, using a fast direct algorithm (Buchberger,
\F4 or \F5), and then computing a Gröbner basis for the wanted order
with either a direct algorithm or a change of order algorithm (Gröbner walk in positive dimension,
FGLM in zero dimension).

The first step of the computation involves choosing a monomial order
making the computations easier.  In the homogeneous case, the usual
choice is the \grevlex order, together with a strategy for selecting
critical pairs for reduction by lowest degree first.  In order to take
advantage from the weighted homogeneous structure of the system $F$,
we may choose the $W$-\grevlex order instead, with a selection
strategy by lowest $W$-degree first.

For algorithms proceeding purely with critical pairs, such as
Buchberger, \F4 or \F5, but unlike Matrix-\F5 for example, this
computation can be performed without changing the algorithm or its
implementation, by transforming the system beforehand:

\begin{prop}
  \label{lemme:passage_par_homW}
  Let $F=(f_{1},\dots,f_{m})$ be a family of polynomials in
  $\K[X_{1},\dots,X_{n}]$, assumed to be weighted homogeneous for a
  system of weights $W = (w_{1},\dots,w_{n})$.  Let $\order{1}$ be a
  monomial order, $G$ be the reduced Gröbner basis of $\hom{W}(F)$ for
  this order, and $\order{2}$ be the pullback of $\order{1}$ through
  $\hom{W}\,$.  Then
  \begin{enumerate}
    \item all elements of $G$ are in the image of $\hom{W}\,$;
    \item the family $G' \coloneq \hom[-1]{W}(G)$ is a reduced Gröbner
    basis of the system $F$ for the order $\order{2}\,$.
  \end{enumerate}
\end{prop}
\begin{proof}
  The morphism $\hom{W}$ preserves $S$-polynomials, in the sense that
  \begin{equation}
    \label{eq:67}
    \Spol(\hom{W}(f),\hom{W}(g)) = \hom{W}(\Spol(f,g)).
  \end{equation}

Recall that we can compute a reduced Gröbner basis by running the Buchberger algorithm, which involves only multiplications, additions, tests of divisibility and computation of $S$-polynomials.
Since all these operations are compatible with $\hom{W}$, if we run the Buchberger algorithm on both $F$ and $\hom{W}(F)$ simultaneously, they will follow exactly the same computations up to application of $\hom{W}$.
The consequences on the final reduced Gröbner basis follow.
\end{proof}

Actually, the fact that $\hom{W}$ preserves $S$-polynomials proves that running any critical pairs algorithm on $\hom{W}(F)$ for the \grevlex order involves exactly the same reductions as running the same algorithm on $F$ for the \Wgrevlex order.

We will only give estimates for the complexity of the \F5~algorithm,
as it is usually faster than Buchberger and \F4.  The complexity of
this algorithm is usually studied through its variant Matrix-\F5.
This complexity is given by
\begin{equation}
  \label{eq:22}
  C_{\F5} = \Ocomp{M_{W,\dreg}^{\omega}}
\end{equation}
where $M_{W,d}$ is the size of the matrix we need to build at
$W$-degree $d$, $\dreg$ is the degree of regularity and $\omega$ is
the exponent in the complexity of matrix multiplication.

For a $W$-homogeneous system, the size of the matrix at $W$-degree $d$
is given by the number of monomials at $W$-degree $d$.
This number of monomials is known as the \emph{Sylvester denumerant} $d(d;w_{1},\dots,w_{n})$.
There is no formula for this denumerant, but its asymptotics are known
(see for example \citet[sec.~4.2]{alfonsin2005}):
\begin{equation}
  \label{eq:48}
  M_{W,d} \simeq \frac{1}{\prod w_{i}} M_{\WOne,d}
  = \frac{1}{\prod w_{i}} \binom{n+d-1}{d}.
\end{equation}
As for the degree of regularity of the system, depending on the
hypotheses satisfied by the system $F$ (regularity, Noether position
or semi-regularity), we can use the corresponding estimates.

All in all, the complexity of computing an ``easy'' Gröbner basis for a
weighted homogeneous system is divided by ${(\prod w_{i})}^{\omega}$
when compared to an homogeneous system with the same degree.  The degree
of regularity is also reduced, yielding an important practical gain
for the \F5~algorithm:
\begin{equation}
  \label{eq:comp-F5}
  C_{\F5}=\Ocomp{
    \frac{1}{\left(\prod w_{i}\right)}^{\omega}
    \cdot \binom{n+ \dreg-1}{\dreg}^{\omega}
  } .
\end{equation}

The gain from the reduced number of monomials applies to other
algorithms as well, provided they are run on $\hom{W}(F)$ if they are
only using critical pairs, or use the \hbox{\Wgrevlex} order otherwise.

Solving zero-dimensional weighted homogeneous systems is rarely
needed.  The reason is that generically, such a system only admits the
trivial solution $(0,\dots,0)$.  For most applications, a \Wgrevlex
Gröbner basis is enough, without the need for a change of ordering.

For positive dimension, depending on the situation, it may be
interesting to perform a two-steps computation, or to simply use one
of the direct algorithms with the desired order.  In the former case,
the usual algorithm used for the change of order is the Gröbner walk.
This algorithm is much more complex and to the best of our knowledge,
does not have good complexity estimates.  However, it involves
computing successive Gröbner bases, using algorithm \F4 or \F5 as a
blackbox.  As such, assigning weights to a polynomial system will
yield similar improvements for the computing time.

\subsection{Affine systems}
\label{sec:Affine-systems}

Affine systems can be solved with the same methods as homogeneous or
weighted homogeneous systems, by homogenizing the system with an
homogenization variable $H$. However, reducing affine systems can lead
to \emph{degree falls}, that is reductions leading to affine
polynomials of lesser $W$-degree, or equivalently, to weighted
homogeneous polynomials divisible by $H$. If the algorithm carries on
the computation on the homogenized system, then it will be led to
examine polynomials divisible by large powers of $H$. This effect can
be mitigated by detecting these reductions and reinjecting these
polynomials at the relevant $W$-degree, but overall, degree falls
usually make the computation slower.

Such a degree fall is a reduction to zero of the highest $W$-degree
components of a pair of polynomials.  However, if the highest
$W$-degree components form a regular sequence (or a sequence in
Noether position, or a sequence in simultaneous Noether position), all
results from the $W$-homogeneous case apply. For semi-regular sequences,
the \F5 Criterion can only eliminate degree falls up to the last
$W$-degree $\delta$ at which all of the multiplication applications
$s_{i,d}$ ($n < i \leq m$) are injective. At this degree, a degree
fall is unavoidable, and the algorithm is left to proceed with the
lower $W$-degree components of the system, for which no regularity
assumption was made. However, the degree of these subsequent
reductions will not go above $\delta$, and complexity estimates can be
obtained by considering the full Macaulay matrix at $W$-degree $\delta$.

Assuming the affine system is zero-dimensional, we may ultimately want to compute its solutions.
This is done by writing triangular generators of the ideal.
Using Gröbner bases, generically, it requires computing a Gröbner basis of the ideal for the lexicographical order, which can be done with a change of order from the \Wgrevlex order.
The usual algorithm for that purpose is the \FGLM algorithm.
Its complexity is given by
\begin{equation}
  \label{eq:50}
  C_{\FGLM}=\Ocomp{n\deg^{\omega}}
\end{equation}
where $\deg$ is the degree of the system.

Let $F=(f_{1},\dots,f_{n}) \subset \KK[X_{1},\dots,X_{n}]$ be a
zero-dimensional affine system.  For any system of weights
$W=(w_{1},\dots,w_{n})$, one may $W$-homogenize the system $F$, that
is compute a system
$F^{h}=(f_{1}^{h},\dots,f_{n}^{h}) \subset \KK[X_{1},\dots,X_{n},H]$
such that for any $i \in \{1,\dots,n\}$,
\begin{equation}
  \label{eq:65}
  f_{i}(X_{1},\dots,X_{n}) = f_{i}^{h}(X_{1},\dots,X_{n},1),
\end{equation}
and such that $F^{h}$ is $W^{h}$-homogeneous, with
$W^{h}=(w_{1},\dots,w_{n},1)$.

If $F$ is regular, then its homogenized $F^{h}$ is also regular.
Assume that the system of weights $W$ is chosen so that $F$ is regular
in the affine sense, i.e. its highest $W$-degree components form a
regular sequence.  Since the system of these highest $W$-degree
components is exactly $F^{h}(H \coloneq 0)$, by the
characterization~\ref{item:NP4}, $F^{h}$ is in Noether position with
respect to the variables $X_{1},\dots,X_{n}$.  As a consequence, the
degree of $\langle F^{h} \rangle$ is
\begin{equation}
  \label{eq:66}
  \deg =\frac{\prod_{i=1}^{n} d_{i}}{\prod_{i=1}^{n} w_{i}}
\end{equation}
and the complexity bounds for the change of ordering are also improved
by a factor ${\left(\prod_{i=1}^{n}w_{i}\right)}^{\omega}$:
\begin{equation}
  \label{eq:comp-FGLM}
  C_{\FGLM}=\Ocomp{ n {\left(\frac{\prod d_{i}}{\prod w_{i}}\right)}^{\omega}}.
\end{equation}



\section{Applications}
\label{sec:Applications}

In this section, we present some applications where taking into
account the weighted structure of the system yields speed-ups.  For
each system, we compare two strategies: the ``standard'' strategy
 consists of computing a Gröbner basis without considering the
weighted structure; the ``weighted'' strategy is the strategy we
described at section~\ref{sec:Cons-comp}. For all these examples, we
use a more compact notation for degrees and weights, so that for example,
$(2^{3},1)$ is equivalent to $(2,2,2,1)$.%

\subsection{Generic systems}
\label{sec:Generic-systems}

First, we present some timings obtained with generic systems, in both the complete intersection ($m=n$), the positive-dimensional ($m<n$) and the over-determined ($m>n$) cases.
In both cases, we fix a system of weights $W=(w_{1},\dots,w_{n})$ and a system of \hbox{$W$-degrees} $D=(d_{1},\dots,d_{m})$, and we pick at random $m$ polynomials $(f_{i})_{i=1\dots m}$, such that for any $i \in \{1,\dots,m\}$, $f_{i}$ has dense support in the set of monomials with $W$-degree less than or equal to $d_{i}$.

For complete intersection systems, we compute a lexicographical
Gröbner basis, using a two-steps strategy in Magma, with algorithm~\F4
as a first step (first block of lines in Table~\ref{tab:gen-DRL}) and
algorithm~\FGLM for the change of ordering (Table~\ref{tab:gen-lex}).

For over-determined systems, we compute a Gröbner basis for the
\grevlex ordering, using algorithm~\F4 from Magma (second block of
lines in Table~\ref{tab:gen-DRL}). 

For positive-dimensional systems, we compute a basis for an elimination order, using a two-steps strategy with FGb%
\footnote{The Gröbner basis algorithms from Magma seem to behave strangely with elimination orders, as seen in the detailed logs, and it coincides with significant slowdowns.
  This behavior was not observed on other implementations of the same algorithms: \F4 from FGb and Buchberger from Singular~\citep{Singular}.
  For example, for the system in the first line of table~\ref{tab:gen-elim}, without the weights, with Magma's \F4 algorithm, the first degree fall comes at step 4, and the algorithm needs more than 66 steps to compute the basis. With FGb's implementation of \F4 in Maple, the first degree fall appears at step 13, and the algorithm finishes at step 32.}%
: first we compute a \grevlex basis with algorithm~\F4 (third block of lines in Table~\ref{tab:gen-DRL}), and then we compute a basis for the wanted elimination order, again with \F4 (Table~\ref{tab:gen-elim}).
In this table, the second column describes what variables we eliminate: for example, $3$ means that we eliminate the first $3$ variables, while $1 \rightarrow 3$ means that we first eliminate the first variable, then the next $2$ variables, again resulting in a basis eliminating the first $3$ variables.

For algorithm~\F4 with the \grevlex ordering, the behavior we observe
is coherent with the previous complexity studies: we observe some
speed-ups when taking into account the weighted structure of the
system, and these speed-ups seem to increase with the
weights. However, the speed-ups cannot be expected to match rigorously
the ones predicted by the complexity bounds, because the systems are
usually not regular for the standard strategy. Experiments also
confirm that it is more effective to order the variables with highest
weight first.

For the lexicographical ordering with \FGLM, we also observe some
speed-ups when applying the weights (we will observe this behavior
again in Section~\ref{sec:Discr-logar-probl}). These differences are
not explained by the theoretical complexity bounds, since both ideals
have the same degree in each case. However, it appears that the slower
\FGLM runs are those where the \FGLM matrix is denser, and that this
difference in density seems to match quantitatively the speed-ups we
observe.

Finally, for elimination bases, the results are similar to what we
observed with the \grevlex ordering: when possible, one should take
into account the weights, and order the variables such that the
smallest weights are also the smallest variables. However, when
eliminating variables, the largest variables need to be the ones that
should be eliminated. If the variables need to be ordered such that
those with the smallest weights are first, in most cases, taking into account the
weighted structure is still profitable. However, if the smallest
weight is on the largest variable and there is only one such variable,
this is no
longer true (see for example the second line in
Table~\ref{tab:gen-elim}). Experiments suggest that these systems
naturally possess a good weighted structure for the weights
$(1,\dots,1)$: their construction ensures that every such polynomial of
total degree $d$ will have a large homogeneous component at degree
$d/2$, and the higher degree components will be small, and divisible
by large powers of $X_{1}$. On the other hand, with weights
$(1,2,\dots,2)$, the same polynomial will have a large $W$-homogeneous
component at $W$-degree $d$, overall leading to reductions at higher
degree (an example is given in Table~\ref{tab:components}).%

\begin{table}
 
  \caption{Benchmarks with Magma for generic systems}
  \label{tab:gen}
  
  \setlength{\abovetopsep}{1\parsep}
  \setlength{\belowbottomsep}{1\parsep}
  \sisetup{
    table-number-alignment=right,
    table-text-alignment=right
  }

  \begin{subtable}[t]{\textwidth}
    \centering
    \footnotesize
    
    \begin{tabular}{lSSS}
      \toprule							
      {Parameters}
      & {\pbox[r]{Without\\weights (\si{\second})}}
      & {\pbox[r]{With\\weights (\si{\second})}}
      & {Speed-up} \\
      \midrule							
      $n=8$, $W=(2^{6},1^{2})$, $D=(4^{8})$    & 8.0                               & 2.5                            & 3.2        \\
      $n=9$, $W=(2^{7},1^{2})$, $D=(4^{9})$    & 101.2                             & 12.5                           & 8.1        \\
      \midrule							
      $n=7$, $W=(2^{5},1^{2})$, $D=(8^{15})$   & 31.6                              & 7.5                            & 4.2        \\
      $n=7$, $W=(2^{5},1^{2})$, $D=(8^{14})$   & 29.0                              & 9.4                            & 3.1        \\
      $n=7$, $W=(2^{5},1^{2})$, $D=(8^{13})$   & 40.0                              & 12.0                           & 3.3        \\
      \midrule							
      $n=5$, $m=4$, $W=(2^{4},1)$, $D=(8^{4})$ & 2.6                               & 0.2                            & 13.0       \\
      $n=5$, $m=4$, $W=(1,2^{4})$, $D=(8^{4})$ & 2.5                               & 0.3                            & 8.3        \\
      $n=5$, $m=4$, $W=(1^3,2^2)$, $D=(4^4)$   & 23.6                              & 0.0                            & 2360.0     \\
      $n=5$, $m=4$, $W=(2^2,1^3)$, $D=(4^4)$   & 407.5                             & 0.0                            & 40750.0    \\
      \bottomrule							
    \end{tabular}

    \caption{Benchmarks for the \F4 algorithm for the \grevlex ordering}
    \label{tab:gen-DRL}
    
  \end{subtable}  

\end{table}

\begin{table}
  \ContinuedFloat
  \begin{subtable}[t]{\textwidth}
    \centering

    \footnotesize

    \begin{tabular}{lSSSS}
      \toprule									
      {Parameters}                          & {Degree}
      & {\pbox[r]{Without\\weights (\si{\second})}}
      & {\pbox[r]{With\\weights (\si{\second})}}
      & {Speed-up} \\
      \midrule									
      $n=8$, $W=(2^{6},1^{2})$, $D=(4^{8})$ & 1024.0   & 500.4                             & 495.0                          & 1.0        \\
      $n=9$, $W=(2^{7},1^{2})$, $D=(4^{9})$ & 2048.0   & 11995.8                           & 7462.1                         & 1.6        \\
      \bottomrule									
    \end{tabular}

    \caption{Benchmarks for the \FGLM algorithm (lexicographical ordering)}
    \label{tab:gen-lex}

  \end{subtable}
\end{table}

\begin{table}
  \ContinuedFloat

  \begin{subtable}[t]{\textwidth}
    \centering

    \footnotesize
    
    \begin{tabular}{ll
        S[table-comparator=true]
        SS[table-comparator=true,table-figures-integer=4]}
      \toprule									
      {Parameters}                             & {Elim. vars.}
      & {\pbox[r]{Without \\ weights (\si{\second})}}
      & {\pbox[r]{With\\weights (\si{\second})}}
      &	{Speed-up} \\
      \midrule									
      $n=5$, $m=4$, $W=(2^{4},1)$, $D=(8^{4})$ & $1$             & 120.3                             & 12.0                           &	10.0       \\
      $n=5$, $m=4$, $W=(1,2^{4})$, $D=(8^{4})$ & $1$             & 27.6                              & 30.4                           &	0.9        \\
      $n=5$, $m=4$, $W=(1^3,2^2)$, $D=(4^4)$   & $2$             & 146.3                             & 6.9                            &  21.2          \\
      $n=5$, $m=4$, $W=(1^3,2^2)$, $D=(4^4)$   & $1 \to 2$       & 162.0                             & 3.3                            &	49.1       \\
      $n=5$, $m=4$, $W=(2^2,1^3)$, $D=(4^4)$   & $1$             & >
      750                             & 0.1
      &  > 7500          \\
      $n=5$, $m=4$, $W=(2^2,1^3)$, $D=(4^4)$   & $1 \to 2$       & NA                                  & 0.1                            &  NA          \\
      $n=5$, $m=4$, $W=(2^2,1^3)$, $D=(4^4)$   & $1 \to 2 \to 3$ & NA                                  & 7.9                            &  NA          \\
      \bottomrule									
    \end{tabular}

    \caption{Benchmarks for the \F4 algorithm for elimination}
    \label{tab:gen-elim}

  \end{subtable}

\end{table}



\begin{table}
  \centering
  \caption{Size of the $W$-homogeneous components for a generic polynomial
    with $W_{0}$-degree $4$ for $W_{0}=(1,2,2,2)$}
  \label{tab:components}

  \begin{tabular}{crrr}
    \toprule
    $W$-degree & $W=(1,2,2,2)$ & $W=(1,1,2,2)$ & $W=(1,1,1,1)$ \\
    \midrule
    0      & 1             & 1             & 1             \\
    1      & 1             & 2             & 4             \\
    2      & 4             & 5             & 10            \\
    3      & 4             & 6             & 4             \\
    4      & 10            & 6             & 1             \\
    \bottomrule
  \end{tabular}
  
 \end{table}

We conclude this section with timings illustrating the consequences of
the estimates of the degree of regularity of a system, depending on
the order of the variables (Section~\ref{sec:Degr-regul-quasi}). For
this purpose, we generate a generic system of $W$-degree $(60^{4})$
with weights $(20,5,5,1)$. Then we compute a \Wgrevlex Gröbner basis
for the orders $X_{1} > \dots > X_{4}$ (smallest weights last) and for
the reverse order $X_{n} < \dots < X_{1}$. We give the degree of
regularity, the value predicted by the previous bound~\eqref{eq:91},
by the new bound~\eqref{eq:1} and by the conjectured
bound~\eqref{eq:130}, as well as the timings. This experiment was run
using algorithm \F5 from the FGb library, the results are in Table~\ref{tab:dreg-timings}.

\begin{table}
  \centering
  
  \caption{Impact of the order of the variables on the degree of
    regularity and the computation times (generic weighted homogeneous system with
    $W$-degree $(60^{4})$ w.r.t. $W=(20,5,5,1)$)}
  \label{tab:dreg-timings}

  \begin{tabular}{crrrrS[table-figures-integer=3,table-figures-decimal=1]}
    \toprule
    Order                     & $\dreg$                          & \pbox{Macaulay's\\ bound~\eqref{eq:91}} & 
    Bound~\eqref{eq:1}        & \pbox{Conjectured\\ bound~\eqref{eq:130}} & {\F5 time}                             \\
    \midrule
    $X_{1}>X_{2}>X_{3}>X_{4}$ & 210                              & 229
                              & 210                              & 210                            & 101.9 \\
    $X_{4}>X_{3}>X_{2}>X_{1}$ & 220                              & 229                            & 229
                              & 220                              & 255.5                                  \\
    \bottomrule
  \end{tabular}

\end{table}

\subsection{Discrete logarithm problem}
\label{sec:Discr-logar-probl}

Taking advantage of a weighted homogeneous structure has allowed the
authors of the article~\citep{FGHR13} to obtain significant speed-ups for solving
a system arising from the DLP on Edwards elliptic curves (\cite{Gaudry2009}).
They observed that the system of equations they had to solve has
symmetries, and rewrote it in terms of the invariants of the symmetry
group.  For a system in $n$ equations, these invariants are
\begin{equation}
  \begin{array}{lcl}
    E_{1}   & =  & e_{1}(X_{1}^{2},\dots,X_{n}^{2})   \\
    E_{2}   & = & e_{2}(X_{1}^{2},\dots,X_{n}^{2})   \\
    & \vdots &                              \\
    E_{n-1} & = & e_{n-1}(X_{1}^{2},\dots,X_{n}^{2}) \\
    E_{n} & = & e_{n}(X_{1},\dots,X_{n}).
  \end{array}
\end{equation}

The system they obtained is sparser, but does not have a good
homogeneous structure. In particular, the highest total degree
components of the system do not form a regular sequence, and in
practice, a Gröbner basis computation will follow many degree falls.

However, the system had a weighted homogeneous structure for the
weights $(2,\dots,2,1)$ (only $E_{n}$ has weight $1$), with respective
$W$-degree $(2^{n},\dots,2^{n})$.  The highest $W$-degree components
forming a sequence in simultaneous Noether position with respect to
the order $E_{1} > E_{2} > \dots > E_{n}$, one could compute a Gröbner
basis without any $W$-degree fall, with complexity bounded by the
estimates~\eqref{eq:comp-F5} and~\eqref{eq:comp-FGLM}.


\begin{table}[t]
  \caption{Benchmarks with FGb and Magma for DLP systems}
  \label{tab:benchmarks-DLP}

  \sisetup{
    table-number-alignment=right,
    table-text-alignment=right
  }

  \begin{subtable}[t]{\textwidth}
    \centering

    \footnotesize
    \begin{tabular}{ l S[table-figures-integer=6,
        table-figures-decimal=0] S[table-figures-integer=3,
        table-figures-decimal=1] S[table-figures-integer=4,
        table-figures-decimal=1] S[table-figures-integer=1,
        table-figures-decimal=1] S[table-figures-integer=4,
        table-figures-decimal=1] S[table-figures-integer=4,
        table-figures-decimal=1] S[table-figures-integer=1,
        table-figures-decimal=1] }
      \toprule
      {System} &
      {$\deg(I)$} &
      {\F5 w (\si{\second})} &
      {\F5 std (\si{\second})} &
      {\pbox[r]{Speed-up\\for \F5}} &
      {\pbox[r]{\FGLM\\w (\si{\second})}} &
      {\pbox[r]{\FGLM\\std (\si{\second})}} &
      {\pbox[r]{Speed-up\\for \FGLM}} \\
      \midrule
      \pbox{DLP Edwards: ${n=4}$, \\${W= (2^{3},1)}$, ${D= (8^4)}$} & 512 & 0.1 & 0.1 & 1.0 & 0.1 & 0.1 & 1.0 \\
      \addlinespace[2.5pt]\pbox{DLP Edwards: ${n=5}$, \\${W= (2^{4},1)}$, ${D= (16^5)}$} & 65536 & 935.4 & 6461.2 & 6.9 & 2164.4 & 6935.6 & 3.2 \\
      \bottomrule
    \end{tabular}

    \caption{Benchmarks with FGb}
    \label{tab:benchmarkFGb}

  \end{subtable}

  \begin{subtable}[t]{\textwidth}
    \centering
    \footnotesize
    \begin{tabular}{ l S[table-figures-integer=5,
        table-figures-decimal=0] S[table-figures-integer=4,
        table-figures-decimal=0] S[table-figures-integer=5,
        table-figures-decimal=0] S[table-figures-integer=1,
        table-figures-decimal=1] S[table-figures-integer=1,
        table-figures-decimal=0] S[table-figures-integer=2,
        table-figures-decimal=0] S[table-figures-integer=2,
        table-figures-decimal=0] }
      \toprule
      {System} &
      {$\deg(I)$} &
      {\F4 w (\si{\second})} &
      {\F4 std (\si{\second})} &
      {\pbox[r]{Speed-up\\for \F4}} &
      {\pbox[r]{\FGLM\\w (\si{\second})}} &
      {\pbox[r]{\FGLM\\std (\si{\second})}} &
      {\pbox[r]{Speed-up\\for \FGLM}} \\
      \midrule
      \pbox{DLP Edwards: ${n=4}$, \\${W= (2^{3},1)}$, ${D= (8^4)}$} & 512 & 1 & 1 & 1.0 & 1 & 27 & 27 \\
      \addlinespace[2.5pt]\pbox{DLP Edwards: ${n=5}$, \\${W= (2^{4},1)}$,
        ${D= (16^5)}$} & 65536 & 6044 & 56105 & 9.3 & {$\infty$} &
      {$\infty$} & {NA} \\
      \bottomrule
    \end{tabular}
    \caption{Benchmarks with Magma}
    \label{tab:benchmarkMagma}

  \end{subtable}

\end{table} 

\subsection{Polynomial inversion}
\label{sec:Polynomial-inversion}

The polynomial inversion problem consists of finding polynomial relations between
polynomials. More precisely, given a system of polynomial
equations
\begin{equation}
  \label{eq:8}
  \left\{ 
  \begin{array}{rcl}
    f_{1}(X_{1},\dots,X_{n}) & = & 0 \\
    f_{2}(X_{1},\dots,X_{n}) & = &0 \\
    & \vdots & \\
    f_{m}(X_{1},\dots,X_{n}) & = & 0\,,
  \end{array}
  \right.
\end{equation}
we want to compute all the relations of the form
\begin{equation}
  \label{eq:9}
  g_{i}(f_{1},\dots,f_{r}) = 0.
\end{equation}

One can compute these relations with Gröbner bases by computing an
elimination ideal: consider the ideal generated by the polynomials
\begin{equation}
  \label{eq:10}
  \begin{array}{rcl}
     T_{1} & -      & f_{1}(X_{1},\dots,X_{n}) \\
     T_{2} & -      & f_{1}(X_{1},\dots,X_{n}) \\
           & \vdots &                          \\
     T_{m} & -      & f_{m}(X_{1},\dots,X_{n})
   \end{array}
 \end{equation}
in $R \coloneq \KK[X_{1},\dots,X_{n},T_{1},\dots,T_{m}]$.  Order $R$
with an elimination order for the variables $X_{1},\dots,X_{n}$, that
is an order such that
\begin{gather}
  \label{eq:12}
  m_{X}(X_{1},\dots,X_{n})m_{T}(T_{1},\dots,T_{m}) <_{\mathrm{elim}}
  m'_{X}(X_{1},\dots,X_{n})m'_{T}(T_{1},\dots,T_{m})
                                               \\[5pt]
  \hspace{1cm}\iff
  \begin{cases}
    m_{X} <_{X} m'_{X}                         \\
    \text{or}                                  \\
    m_{X} = m'_{X} \text{ and } m_{T} <_{T} m'_{T}
  \end{cases}
\end{gather}
for some monomial orders $<_{X}$ and $<_{T}$.  The usual choice is a
block-\grevlex order.

This problem can benefit from being given a weighted structure
(see~\cite[sec.~6.1]{Tra1996}).  For any $i \in \{1,\dots,m\}$, let
$d_{i}$ be the degree of $f_{i}$.  By setting the weight of $T_{i}$ to
be $d_{i}$, the monomial $T_{i}$ becomes part of the highest
$W$-degree component of $T_{i} - f_{i}(X_{1},\dots,X_{n})$, giving
this equation a weighted homogeneous structure.

More precisely:
\begin{prop}
  Let $f_{1},\dots,f_{m}$ be a system of polynomials with respective degree $d_{1},\dots,d_{m}$ in $\KK[X_{1},\dots,X_{n}]$.
  Consider the algebra $R \coloneq \KK[X_{1},\dots,X_{n},T_{1},\dots,T_{m}]$, graded with the weights $W=(1,\dots,1,d_{1},\dots,d_{m})$, and consider the system $F=(T_{1}-f_{1}\mathbf{X},\dots,T_{m}-f_{m}(\mathbf{X}))$ in $R$.
  Then the system $F^{h}$ formed with the highest $W$-degree components of $F$ is in Noether position with respect to the variables $T_{1},\dots,T_{m}$, and in particular it forms a regular sequence.
\end{prop}
\begin{proof}
  By the choice of the weights, the system $F^{h}$ is defined by
  \begin{equation}
    \label{eq:16}
    F^{h} = (T_{1}-f^{h}_{1}(\mathbf{X}),\dots,T_{m}-f^{h}_{m}(\mathbf{X})),
  \end{equation}
  where for any $i \in \{1,\dots,m\}$, $f^{h}_{i}$ is the highest
  degree component of $f_{i}$.  As a consequence, by the
  characterization~\ref{item:NP4} of the Noether position, the system
  $F^{h}$ is indeed in Noether position with respect to the variables
  $T_{1},\dots,T_{m}$.
\end{proof}


In Tables~\ref{table:benchmarks-polynomial-inversion}, we present
timings for a few systems with this kind of problem:
\begin{itemize}
  \item group invariants (\cite{Sturmfels2008}): given a group, compute its fundamental
  invariants, and then the relations between these invariants. Since
  these examples can lead to very long computations, in some cases, we
  only compute the relations between the $k$ first invariants;
  \item monomials: given $m$ monomials of degree $d$ in
  $\KK[X_{1},\dots,X_{n}]$, compute the relations between them;
  \item matrix minors: given a $p \times q$ matrix of linear forms in
  $n$ indeterminates, compute all its minors of rank $r$ as
  polynomials in the $X_{i,j}$'s, and compute the relations between
  them.
\end{itemize}

In each case, we compute an elimination basis using a two-steps
strategy: first we compute a \grevlex basis (Table~\ref{tab:inv-DRL}),
then we compute the elimination basis
(Table~\ref{tab:inv-DRL-to-elim}). In Table~\ref{tab:inv-elim}, we
show some timings for the computation of the elimination basis
directly from the input system. All these experiments were run using
algorithm \F4 from Magma.


\begin{table}
  \caption{Benchmarks with Magma on some polynomial inversion systems}
  \label{table:benchmarks-polynomial-inversion}

  \setlength{\abovetopsep}{1\parsep}
  \setlength{\belowbottomsep}{1\parsep}
  \begin{subtable}[t]{\textwidth}
    \centering
    \footnotesize
    \begin{tabular}{
        l
        S[table-figures-integer=6,table-figures-decimal=1,
        table-space-text-post= \textsuperscript{a},
        table-comparator=true,
        table-number-alignment=right]
        S[table-figures-integer=3,table-figures-decimal=2]
        S[table-figures-integer=2,table-figures-decimal=1]
      }
      \toprule							
      {System}
      & {\pbox[r]{Without\\weights (\si{\second})}}
      & {\pbox[r]{With\\weights (\si{\second})}}
      & {Speed-up} \\
      \midrule							
      Cyclic invariants, $n=4$                                     & 4.2                               & 0.0                            & 140.0      \\
      Cyclic invariants, $n=5$, $k=12$                             & 2612.6                            & 54.7                           & 47.8       \\
      Cyclic invariants, $n=5$                                     & >75000\textsuperscript{a}         & 392.7                          & NA         \\
      Cyclic invariants, $n=6$, $k=14$                             & 32987.6                           & 2787.7                         & 11.8       \\
      Cyclic invariants, $n=6$, $k=15$                             & >280000\textsuperscript{a}        & 14535.4                        & NA         \\
      Dihedral invariants, $n=5$                                   & >70000\textsuperscript{a}         & 6.3                            & NA         \\
      \midrule							
      Generic monomials, $d=2$, $n=24$, $m=48$                     & 216.1                             & 0.2                            & 1350.6     \\
      Generic monomials, $d=2$, $n=25$, $m=50$                     & 14034.7                           & 0.1                            & 116955.8   \\
      Generic monomials, $d=2$, $n=26$, $m=52$                     & 14630.6                           & 0.2                            & 66502.7    \\
      Generic monomials, $d=2$, $n=27$, $m=54$                     & 8887.6                            & 0.2                            & 55547.5    \\
      Generic monomials, $d=3$, $n=11$, $m=22$                     & 370.9                             & 0.1                            & 6181.7     \\
      Generic monomials, $d=3$, $n=12$, $m=24$                     & 4485.0                            & 0.2                            & 26382.4    \\
      \midrule							
      Matrix minors, $n=5$, $7 \times 7$, $r=3$                    & 125.7                             & 93.3                           & 1.3        \\
      Matrix minors, $n=6$, $7 \times 7$, $r=3$                    & 1941.0                            & 1029.1                         & 1.9        \\
      Matrix minors, $n=6$, $8 \times 8$, $r=3$                    & 4115.8                            & 2295.8                         & 1.8        \\
      Matrix minors, $n=4$, $6 \times 6$, $r=5$                    & 612.6                             & 159.2                          & 3.8        \\
      Matrix minors, $n=4$, $7 \times 7$, $r=6$                    & 8043.3                            & 2126.9                         & 3.8        \\
      Matrix minors, $n=4$, $7 \times 10$, $r=7$                   & 69386.1                           & 43910.1                        & 1.6        \\
      \bottomrule							
      \footnotesize{a. Memory usage was over \SI{120}{\giga\byte}} &                                   &                                & 
    \end{tabular}

    \caption{First step (\F4 for the \grevlex order)}
    \label{tab:inv-DRL}
    
  \end{subtable}

\end{table}

\begin{table}
  \ContinuedFloat

  \begin{subtable}[t]{\textwidth}
    \centering
    \footnotesize
    \begin{tabular}{
        l
        S[table-figures-integer=5,table-figures-decimal=1]
        S[table-figures-integer=5,table-figures-decimal=1]
        S[table-figures-integer=2,table-figures-decimal=1]
      }
      \toprule							
      {System}
      & {\pbox[r]{Without\\weights (\si{\second})}}
      & {\pbox[r]{With\\weights (\si{\second})}}
      &	{Speed-up} \\
      \midrule							
      Cyclic invariants, $n=4$                   & 7.0                               & 0.1                            &	70.0       \\
      Cyclic invariants, $n=5$, $k=12$           & 1683.2                            & 70.7                           &	23.8       \\
      Cyclic invariants, $n=5$                   & NA                                & 382.5                          &	NA         \\
      Cyclic invariants, $n=6$, $k=14$           & 9236.4                            & 1456.0                         &	6.3        \\
      Cyclic invariants, $n=6$, $k=15$           & NA                                & 7179.7                         &	NA         \\
      Dihedral invariants, $n=5$                 & NA                                & 20.3                           &	NA         \\
      \midrule							
      Generic monomials, $d=2$, $n=24$, $m=48$   & 250.3                             & 117.4                          &	2.1        \\
      Generic monomials, $d=2$, $n=25$, $m=50$   & 13471.2                           & 15932.9                        &	0.8        \\
      Generic monomials, $d=2$, $n=26$, $m=52$   & 17599.5                           & 8054.2                         &	2.2        \\
      Generic monomials, $d=2$, $n=27$, $m=54$   & 9681.0                            & 3605.6                         &	2.7        \\
      Generic monomials, $d=3$, $n=11$, $m=22$   & 624.5                             & 199.9                          &	3.1        \\
      Generic monomials, $d=3$, $n=12$, $m=24$   & 9751.6                            & 3060.1                         &	3.2        \\
      \midrule							
      Matrix minors, $n=5$, $7 \times 7$, $r=3$  & 52.6                              & 66.6                           &	0.8        \\
      Matrix minors, $n=6$, $7 \times 7$, $r=3$  & 556.5                             & 779.1                          &	0.7        \\
      Matrix minors, $n=6$, $8 \times 8$, $r=3$  & 1257.9                            & 1714.0                         &	0.7        \\
      Matrix minors, $n=4$, $6 \times 6$, $r=5$  & 262.7                             & 328.1                          &	0.8        \\
      Matrix minors, $n=4$, $7 \times 7$, $r=6$  & 2872.2                            & 4299.8                         &	0.7        \\
      Matrix minors, $n=4$, $7 \times 10$, $r=7$ & 4728.4                            & 5485.8                         &	0.9        \\
      \bottomrule							
    \end{tabular}

    \caption{Second step (\F4 for an elimination order)}
    \label{tab:inv-DRL-to-elim}
    
  \end{subtable}
\end{table}

\begin{table}
  \ContinuedFloat
 
  \begin{subtable}[t]{\textwidth}
    \centering
   \footnotesize
    \begin{tabular}{
        l
        S[table-figures-integer=6,
        table-figures-decimal=1,
        table-comparator=true,
        table-space-text-post= \textsuperscript{b},
        table-number-alignment=right
        ]
        S[table-figures-integer=6,table-figures-decimal=1]
        S[table-figures-integer=4,table-figures-decimal=1,table-comparator=true]
      }
      \toprule							
      {System}
      & {\pbox[r]{Without\\weights (\si{\second})}}
      & {\pbox[r]{With\\weights (\si{\second})}} & {Speed-up} \\
      \midrule							
      Cyclic invariants, $n=4$                                      & 4.0                               & 0.3                            & 13.3       \\
      Cyclic invariants, $n=5$, $k=12$                              & 2705.8                            & 73.4                           & 36.9       \\
      Cyclic invariants, $n=5$                                      & >90000\textsuperscript{b} 	& 370.0                          & >243       \\
      Cyclic invariants, $n=6$, $k=14$                              & 35922.4                           & 2256.2                         & 15.9       \\
      Cyclic invariants, $n=6$, $k=15$                              & >300000\textsuperscript{b} 	& 7426.7                         & >40        \\
      Dihedral invariants, $n=5$                                    & >40000\textsuperscript{b} 	& 18.5                           & >2162      \\
      \midrule							
      Generic monomials, $d=2$, $n=24$, $m=48$                      & 216.5                             & 110.9                          & 2.0        \\
      Generic monomials, $d=2$, $n=25$, $m=50$                      & 31135.2                           & 16352.2                        & 1.9        \\
      Generic monomials, $d=2$, $n=26$, $m=52$                      & 14919.2                           & 8142.8                         & 1.8        \\
      Generic monomials, $d=2$, $n=27$, $m=54$                      & 5645.8                            & 4619.0                         & 1.2        \\
      Generic monomials, $d=3$, $n=11$, $m=22$                      & 370.1                             & 193.1                          & 1.9        \\
      Generic monomials, $d=3$, $n=12$, $m=24$                      & 4527.2                            & 2904.6                         & 1.6        \\
      \midrule							
      Matrix minors, $n=7$, $7 \times 7$, $r=3$                     & 41220.0                           & 26340.0                        & 1.6        \\
      Matrix minors, $n=7$, $8 \times 8$, $r=3$                     & 48000.0                           & 18060.0                        & 2.7        \\
      Matrix minors, $n=8$, $8 \times 8$, $r=3$                     & 711690.0                          & 390235.0                       & 1.8        \\
      Matrix minors, $n=4$, $6 \times 6$, $r=5$                     & 613.9                             & 325.4                          & 1.9        \\
      Matrix minors, $n=4$, $7 \times 7$, $r=6$                     & 8059.4                            & 3955.5                         & 2.0        \\
      Matrix minors, $n=4$, $7 \times 10$, $r=7$                    & 71067.8                           & 32721.5                        & 2.2        \\
      \bottomrule							
      \footnotesize{b. Memory usage was over \SI{120}{\giga\byte}}. &                                   &                                & 

    \end{tabular}

    \caption{Direct strategy}
    \label{tab:inv-elim}
    
  \end{subtable}
\end{table}



\clearpage 

\bibliographystyle{elsart-harv}



\end{document}
